%% file: main.tex
\documentclass[12pt,draftcls,onecolumn]{IEEEtran}
\IEEEoverridecommandlockouts
\usepackage{balance}
\usepackage{graphicx}
\usepackage{subfigure}
\usepackage{verbatim}
\usepackage{amsmath}
\usepackage{amssymb}
\usepackage{amsfonts}
\usepackage{dsfont}
\usepackage{bm}
\usepackage{algorithmicx}
\usepackage{algpseudocode}
\usepackage{algorithm}
\usepackage{textcomp}
\usepackage{cite}
\usepackage{color}
\newtheorem{theorem}{\textbf{Theorem}}[section]
\newtheorem{proposition}[theorem]{\textbf{Proposition}}%[section]
\newtheorem{corollary}[theorem]{\textbf{Corollary}}%[section]
\newtheorem{assumption}[theorem]{\textbf{Assumption}}%[section]
\newtheorem{lemma}[theorem]{\textbf{Lemma}}%[section]

\newtheorem{remark}[theorem]{\textbf{Remark}}%[section]

\def\EndProof{\indent\vrule height5pt width6pt depth1pt\hfil\par\medbreak}
\usepackage{lineno}
    \newcommand*\patchAmsMathEnvironmentForLineno[1]{%
      \expandafter\let\csname old#1\expandafter\endcsname\csname #1\endcsname
      \expandafter\let\csname oldend#1\expandafter\endcsname\csname end#1\endcsname
      \renewenvironment{#1}%
         {\linenomath\csname old#1\endcsname}%
         {\csname oldend#1\endcsname\endlinenomath}}%
    \newcommand*\patchBothAmsMathEnvironmentsForLineno[1]{%
      \patchAmsMathEnvironmentForLineno{#1}%
      \patchAmsMathEnvironmentForLineno{#1*}}%
    \AtBeginDocument{%
    \patchBothAmsMathEnvironmentsForLineno{equation}%
    \patchBothAmsMathEnvironmentsForLineno{align}%
    \patchBothAmsMathEnvironmentsForLineno{flalign}%
    \patchBothAmsMathEnvironmentsForLineno{alignat}%
    \patchBothAmsMathEnvironmentsForLineno{gather}%
    \patchBothAmsMathEnvironmentsForLineno{multline}%
    }
\linenumbers

\title{\bf
Chebyshev Polynomials in Distributed Consensus Applications
}

\author{\authorblockN{Eduardo Montijano, Juan I. Montijano, and Carlos Sagues}%
\thanks{E. Montijano and C. Sagues are with
Departamento de Inform\'{a}tica e Ingenier\'{\i}a de Sistemas - Instituto de Investigaci\'{o}n en Ingenier\'{\i}a de Arag\'{o}n (I3A),
Universidad de Zaragoza, Spain. {\tt\small emonti@unizar.es, \tt\small csagues@unizar.es}}
\thanks{J.I. Montijano is with Departamento de Matem\'{a}tica Aplicada - %IUMA,
Instituto Universitario de Matem\'{a}ticas y Aplicaciones (IUMA),
    Universidad de Zaragoza, Spain.
    {\tt\small monti@unizar.es}}%
}

\begin{document}
\maketitle

\thispagestyle{empty}
\pagestyle{empty}

\begin{abstract}
In this paper we analyze the use of Chebyshev polynomials
in distributed consensus applications.
We study the properties of these polynomials
to propose a distributed algorithm
that reaches the consensus in a fast way.
The algorithm is expressed in the form of a linear iteration
and, at each step, the agents only require to
transmit their current state to their neighbors.
The difference with respect to previous
approaches is that the update rule used by the network
is based on the second order difference equation that
describes the Chebyshev polynomials of first kind.
As a consequence,
we show that our algorithm achieves
the consensus using far less iterations than other approaches.
We characterize the main properties of the algorithm
for both, fixed and
switching communication topologies.
The main contribution of the paper is the study
of the properties of the Chebyshev polynomials
in distributed consensus applications,
proposing an algorithm
that increases the convergence rate with respect to existing approaches.
Theoretical results, as well as experiments with synthetic data, show the
benefits using our algorithm.
%This improvement reduces the number of messages
%between the agents, saving both power and time to the networked system.
\end{abstract}

\IEEEpeerreviewmaketitle

%%%%%%%%%%%%%%%%%%%%%%%%%%%%%%%%%%%%%%%%%%%%%%%%%%%%%%%%%%%%%%%%%%%%%%%%%%%%%%%%%%%%%%%%%%%%%%%%
\section{Introduction}
\label{introduction}
\input{introduction}

%%%%%%%%%%%%%%%%%%%%%%%%%%%%%%%%%%%%%%%%%%%%%%%%%%%%%%%%%%%%%%%%%%%%%%%%%%%%%%%%%%%%%%%%%%%%%%%%
\section{\bf Background on Chebyshev Polynomials and Distributed Consensus}
\label{consensus}
\input{consensus}

\section{\bf Consensus algorithm using Chebyshev polynomials}
\label{algorithm}
\input{algorithm}

\section{\bf Analysis with a Fixed Communication Topology}
\label{polynomials}
\input{polynomials2}

%\section{\bf Analysis with Fixed Communication Topology}
%\label{fixedTopology}
%\input{fixedTopology}

\section{\bf Analysis with a Switching Communication Topology}
\label{switching}
\input{switching}

\section{\bf Simulations}
\label{simulations}
\input{simulations2}

\section{\bf Conclusions}
\label{conclusions}
In this paper we have analyzed the properties
of Chebyshev polynomials to design a
fast distributed consensus algorithm.
%In this paper
%We have presented a new
%distributed consensus algorithm using
%Chebyshev polynomials.
We have shown that the proposed algorithm significantly reduces the
number of communication rounds required by the
network to achieve the consensus.
We have provided a theoretical analysis of the
properties of the algorithm
in both fixed and switching communication topologies.
%convergence speed and the properties of the algorithm.
We have also evaluated our method with
an extensive set of simulations.
Both theoretical and empirical analysis show
the goodness of our proposal.
%In addition, we have shown that the evaluation
%of the minimal polynomial in large networks
%can be badly affected by numerical errors.

\section{\bf Acknowledgments}
This work was supported by the project DPI2009-08126
and grant AP2007-03282 Ministerio de Educacion y Ciencia.

%%%%%%%%%%%%%%%%%%%%%%%%%%%%%%%%%%%%%%%%%%%%%%%%%%%%%%%%%%%%%%%%%%%%%%%%%%%%%%%%%%%%%%%%%%%%%%%%

%\balance
%%%%%%%%%%%%%%%%%%%%%%%%%%%%%%%%%%%%%%%%%%%%%%%%%%%%%%%%%%%%%%%%%%%%%%%%%%%%%%%%%%%%%%%%%%%%%%%%
%\bibliographystyle{plain}
\bibliographystyle{unsrt}
\bibliography{Bibliography}

\appendix
\subsection{Proof of Theorem \ref{ConvergenceTh}}
\input{proof_convergenceFixed}
\subsection{Proof of Theorem \ref{Th_optimalParameters}}
\input{proof_optimalParameters}
%\subsection{Proof of Theorem \ref{Th_convFaster}}
%\input{proof_convergenceFaster}
\subsection{Proof of Theorem \ref{ConvergeSwitchingTopTh}}
\input{proof_convergenceSwitching}

\end{document}

%% file: introduction.tex
Chebyshev polynomials~\cite{JCM-DCH:02} are a powerful
mathematical tool that has proven to be very helpful
in many different fields of science.
To name a few, they are used in
the modeling of complex chemical reaction systems~\cite{CN-HHC-AMD:02},
the simulation satellite orbits
around the Earth.~\cite{DLR-DS-JM:98},
the numerical solution of diffusion-reactions equations with severely stiff reaction terms~\cite{JGV-BPS:04}
or the recognition of patterns in images using Support Vector Machine classification~\cite{SO-CHC-HAC:11}.
In this paper we study the use of these polynomials in the
field of distributed consensus applications.
%presenting algorithms that require a fewer communications than other approaches.
%Several reasons motivate the study of these polynomials in the context
%of distributed consensus problems.

%Developments of wireless communications and
%small computational devices favor the appearance of
%large size sensor networks and multi-agent systems.
%Distributed algorithms are essential in these kind of systems
%in order to make them scalable and robust.
%Among them, the distributed consensus methods
%are of special interest in perception tasks.
In sensor networks and multi-agent systems,
the consensus problem consists of
making the whole group of agents to reach a common estimation
about a specific measurement.
Within the control community
many different distributed solutions have been proposed in the past years
\cite{FB-JC-SM:09,WR-RWB:08,ROS-RMM:04,AJ-JL-ASM:03,MZ-SM:10,MM-DS-JP-SHL-RMM:07,LX-SB-SL:05}.
%compute a common estimatio
%A common estimation of the world can be computed
%in a distributed fashion using
%these algorithms, see e.g.,
%\cite{FB-JC-SM:09,WR-RWB:08,ROS-RMM:04,AJ-JL-ASM:03,MZ-SM:10,MM-DS-JP-SHL-RMM:07,LX-SB-SL:05}.
It is well known that the number of messages
required to achieve the consensus depends on the network connectivity.
Interesting analysis of convergence have been done in~\cite{GFT-AB-MG:06,EL-FG-SZ:10},
where consensus methods have been shown to behave
in a similar manner as heat differential equations and electrical resistive networks respectively.
Other interesting approaches analyze the convergence with
stochastic link failures~\cite{SP-BB-AEA:10},
switching random networks~\cite{JZ-QW:09}
and asynchronous consensus~\cite{MZ-SM:08b}.
When the size of the network is large,
communications between different pairs of agents become more difficult
due to distance and power constraints.
Under these circumstances the number of iterations
required to reach the consensus is also large.
For that reason a lot of research has been devoted to mitigate this problem,
providing a variety of solutions that reduce the time to
achieve the consensus.

Some works present
continuous-time solutions to achieve consensus in finite time
using non linear methods~\cite{FJ-LW:09,JC:06,LW-FX:10}.
The use of numerical integrators affects the number of iterations
in these approaches
because they depend on the number of steps taken by the method.
The approach in~\cite{CKK-XG:09} proposes a link scheduling
that reaches the consensus in a finite number of steps.
However, in wireless networks, communications of direct neighbors
depend on the distance that separates them and therefore,
there might be situations in which this method cannot be used
because not all the links are feasible.
Other approaches speed up convergence by
sending additional information in the messages.
Following this idea a multi-hop protocol is presented in~\cite{ZJ-RMM:06}
and second order neighbors are considered in~\cite{DY-SX-HZ-YC:10}.
Unfortunately, the amount of additional information in both cases depends on the topology.
This implies that there might be situations in which
large messages must be sent.

The design of the adjacency matrix has been the focus of several works.
For instance,
the work in~\cite{LX-SB:04} provides the optimal weights for the matrix,
as well as good approximations that do not require any
global knowledge about the network topology.
Different algorithms to solve the optimization problem
of finding the best matrix are proposed in~\cite{YK-DWG-IP:09}.
Another optimization method is proposed in~\cite{BJ-MJ:08},
in this case considering a shift-registers method with a fixed gain.
These approaches indeed improve the convergence speed, nevertheless,
they can still be combined with additional techniques in order to accelerate even
more the consensus.
%For example, by sending additional information in the messages.
%convergence speed can also be improved.
%Following this idea a multi-hop protocol is presented in~\cite{ZJ-RMM:06}
%and second order neighbors are considered in~\cite{DY-SX-HZ-YC:10}.
%Unfortunately, the amount of additional information in both cases depends on the topology.
%This implies that there might be situations in which
%large messages must be sent.

The distributed evaluation of polynomials,
as well as the use of previous information in the algorithm,
have turned out to be easy ways to speed up the consensus,
also keeping the good properties found in standard methods.
The minimal polynomial of the adjacency matrix
is used in~\cite{SS-CNH:07b} and~\cite{YY-GBS-LS-JG:09}.
%In these approaches,
Once this polynomial is known, the network can achieve the consensus
in a finite number of communication rounds.
Unfortunately, when the topology of the network is time-varying
this algorithm does not work and for large networks the
computation of the polynomial can be inefficient.
The approach in~\cite{EK-PF:09} uses a polynomial of fixed degree
with coefficients computed assuming the network is known.
A consensus predictor is considered in~\cite{TCA-BNO-MJC:09}.
Different second order recurrences with fixed gains
are used in~\cite{BO-MC-MR:10,SM-BG-MHS:98}.
Finally, the distributed evaluation of Chebyshev polynomials
for consensus has been
proposed in~\cite{EM-JIM-CS:11,RLGC-AR-NRJ:11}.
Although the convergence of some of these algorithms
under switching topologies has been demonstrated in practice,
to the authors' knowledge there is still a gap in the theoretical analysis
of the behavior of polynomial evaluation in this case.

In this paper we try to fill this gap, extending the results presented in~\cite{EM-JIM-CS:11}
%We provide a more detailed analysis of the
about Chebyshev polynomials and their use in consensus applications.
%Our algorithms only require to transmit their current state to their neighbors.
%The algorithm, initially introduced in~\cite{EM-JIM-CS:11},
%is based on a second order difference equation and,
%at each step, the nodes only require to
%transmit their current state to their neighbors.
In~\cite{EM-JIM-CS:11} we introduced the algorithm,
based on a second order difference equation, and
%using not only the current local status but also the previous one.
%In~\cite{EM-JIM-CS:11}
we studied its convergence to consensus for
stochastic symmetric matrices in fixed graphs.
In this paper we extend the convergence result,
considering non-symmetric matrices that can have complex eigenvalues.
We also provide a complete study of the parameters that
make the algorithm achieve the optimal convergence rate
and we give bounds
on the selection of these parameters
to achieve a faster convergence than using the powers of the weighted adjacency matrix.
Regarding the case of switching communication topologies,
we are able to theoretically show that there always exist parameters that
make the proposed algorithm converge to the consensus.
Experiments with synthetic data show the
benefits of using our algorithm compared to other methods.

The structure of the paper is the following:
In section II we introduce some background
about the Chebyshev polynomials and distributed consensus.
In section III we present the new distributed consensus algorithm
using Chebyshev polynomials.
In sections IV and V we study the properties of the algorithm
with fixed and switching communication topologies respectively.
In section VI we analyze the behavior of the
algorithm in a simulated setup.
Finally in section VII the conclusions of the work
are presented.
In order to simplify the reading of the manuscript
we have moved to an appendix some of the proofs
of the theoretical results in sections III and IV.
We have left in the text only the proofs that contain
convenient information to follow the analysis. 

%% file: consensus.tex
In this paper we consider
Chebyshev polynomials of the first kind~\cite{JCM-DCH:02}.
We denote the Chebyshev polynomial of degree $n$ by $T_n(x).$
These polynomials satisfy
\begin{equation}
\label{Property1}
T_n(x)= \cos(n \arccos x),\hbox{ for all }x\in[-1,1],
\end{equation}
and $|T_n(x)| >1$ when $|x| >1$, for all $n\in\mathbb{N}.$
A more general way to define these polynomials in the real domain is
%They can also be defined
using a second order recurrence,
\begin{equation}
\label{ChebyshevSimple}
\begin{array}{l}
T_0(x)= 1,\ T_1(x) = x\\
T_n(x)= 2 x T_{n-1}(x) - T_{n-2}(x), \quad n\geq 2.
\end{array}
\end{equation}
By the theory of difference equations~\cite{RPA:92},
the direct expression of \eqref{ChebyshevSimple} is determined by the roots
$\tau_1$ and $\tau_2$ of the characteristic equation,
%From \eqref{ChebyshevSimple},
%another direct expression of $T_n(x)$ can be obtained
%in terms of the roots, $\tau(x)$ and $1/\tau(x),$ of the characteristic equation
%of the homogeneous difference equation, $t^2-2xt-1=0.$
\begin{equation}
\label{direct_Expression2}
T_n(x)= \dfrac{1}{2}(\tau_1(x)^n + \tau_2(x)^{n}), %= \dfrac{1}{2}(\tau(x)^n + \tau(x)^{-n})=\dfrac{1}{2} \tau(x)^{-n}(1+\tau(x)^{2n}),
\end{equation}
where $\tau_1(x)=x-\sqrt{x^2-1}$ and $\tau_2(x)=x+\sqrt{x^2-1}=1/\tau_1(x)$.
In the paper we take
\begin{equation}
\tau(x) = \left\{
                   \begin{array}{ll}
                     x - \sqrt{x^2 -1}, & \hbox{if } x \ge 0\\
                     x + \sqrt{x^2 -1}, & \hbox{if } x < 0
                   \end{array}
                 \right. ,
\end{equation}
so that $|\tau(x)| < 1$ and $|\tau(x)|^{-1} >1$ for all $|x|>1$,
and therefore,
\begin{equation}
\label{direct_Expression}
T_n(x)= \dfrac{1}{2}(\tau(x)^n + \tau(x)^{-n})=\dfrac{1}{2} \tau(x)^{-n}(1+\tau(x)^{2n}).
\end{equation}
It is clear that if $|x|>1$,
then $T_n(x)$ goes to infinity as $n$ grows.
If $|x|<1,$ then $\tau(x)$ is a complex number with $|\tau(x)|=1$
and $|T_n(x)| \le 1,\ \forall n,$
as stated in eq. \eqref{Property1}.

For the analysis in the paper,
it is also convenient to describe the behavior of Chebyshev polynomials
evaluated in complex numbers.
For any  $z \in \mathbb{C},$
Chebyshev polynomials, $T_n(z)$, on the complex plane can also be
expressed by \eqref{direct_Expression}
%\begin{equation*}
%T_n(z)= \dfrac{1}{2}(\tau(z)^n + \tau(z)^{-n})=\dfrac{1}{2} \tau(z)^{-n}(1+\tau(z)^{2n}),
%\end{equation*}
where $\tau(z)$ is defined now by
\begin{equation}
\tau(z) = \left\{
                   \begin{array}{ll}
                     z - \sqrt{z^2 -1}, & \hbox{if }  |z - \sqrt{z^2 -1}| < 1\\
                     z + \sqrt{z^2 -1}, & \hbox{otherwise}
                   \end{array}
                 \right. ,
\end{equation}
and again $|\tau(z)| \le 1$ and $|\tau(z)|^{-1} \ge 1$ for all $z$.
However, note that Chebyshev polynomials evaluated in a complex number, $T_n(z),$
go always to infinity as $n$ grows.

Consider now a set of $N$ agents, $\mathcal{V} = \{1, \dots, N\},$
with limited communication capabilities.
A distributed algorithm achieves consensus if,
starting with initial conditions $x_i(0) \in \mathbb{R},$
and using only local interactions between agents,
$x_i(n) = x_j(n), \forall i,j\in\mathcal{V},$ as $n\to\infty.$
The interactions between the agents are modeled
using an undirected graph
$\mathcal{G}=\{\mathcal{V},\mathcal{E}\}$,
where $\mathcal{E} \subset \mathcal{V} \times \mathcal{V}$ describes the
communications between pairs of agents.
In this way, agents $i$ and $j$ can communicate if and only if $(i,j) \in \mathcal{E}$.
The neighbors of one agent $i\in\mathcal{V}$ are the subset of agents
that can directly communicate with it;
i.e., $\mathcal{N}_i = \{j \in \mathcal{V} \;|\; (i,j) \in \mathcal{E}\}$.
Initially, let us assume that the communication graph is fixed and connected.
%that is, there exists a path of one or more links between any two agents
%in the network.
%We treat the case of switching graphs in section \ref{switching}.

The discrete time distributed consensus algorithm
based on the weighted adjacency matrix associated to the communication
graph~\cite{FB-JC-SM:09} is
\begin{equation}
\label{Consensus_Laplacian_Individual}
x_i(n) = a_{ii} x_i(n-1) + \sum_{j\in\mathcal{N}_i} a_{ij}x_j(n-1),
\end{equation}
with $x_i(0) = x_i.$
The algorithm can also be expressed in vectorial form as
\begin{equation}
\label{Consensus_Laplacian}
%\textbf{x}(n+1) = \textbf{x}(n) +  \textbf{L} \textbf{x}(n) =\textbf{A} \textbf{x}(n),
\textbf{x}(n) = \textbf{A} \textbf{x}(n-1),
\end{equation}
where $\textbf{x}(n) = (x_1(n),\ldots,x_N(n))^T$ and
$\textbf{A}=[a_{ij}] \in \mathbb{R}^{N \times N}$,
is the weighted matrix.
\begin{assumption}[Stochastic Weights]
\label{RowStochastic}
$\textbf{A}$ is row stochastic
and compatible with the underlying graph, $\mathcal{G}$,
i.e., it is such that $a_{ii}\neq0,$ %\ a_{ij}=0$ if $(i,j)\not\in \mathcal{E},$
$a_{ij}\neq0$ only if  $(i,j)\in \cal{E}$
and $\textbf{A} \textbf{1} = \textbf{1}.$
\end{assumption}

Since the communication graph is connected,
by Assumption \ref{RowStochastic},
$\textbf{A}$ has one eigenvalue $\lambda_1=1$
with associated right eigenvector $\textbf{1}$
and algebraic multiplicity equal to one.
The rest of the eigenvalues, real or complex,
satisfy $|\lambda_i| < 1,\ i=2,\ldots,N.$
Without loss of generality, let us suppose that all the eigenvalues are simple.
We denote by $\lambda_2$ the second largest
and $\lambda_N$ the smallest real eigenvalues
and we assume that
$\max\{|\lambda_2|,|\lambda_N|\} > |\lambda_i|,\ i=3,\ldots,N-1.$
%The rest of the eigenvalues, real or complex,
%sorted in decreasing order, % by modulus,
%satisfy $1 > \lambda_2 \ge \ldots \ge \lambda_N > -1$.

Any initial conditions $\textbf{x}(0)$ can be expressed as a sum of eigenvectors of $\textbf{A},$
\[
\textbf{x}(0)= \textbf{v}_1 + \ldots + \textbf{v}_N,
\]
where $\textbf{v}_i$ is a right eigenvector associated to the eigenvalue $\lambda_i$.
Specifically, $\textbf{v}_1$ will be of the form $(\textbf{w}_1^T\textbf{x}(0)/\textbf{w}_1^T\textbf{1})\textbf{1},$
with $\textbf{w}_1$ a left eigenvector of $\textbf{A}$ associated to $\lambda_1.$
It is clear that
\[
\textbf{x}(n)= \textbf{A}^n \textbf{x}(0)=
\textbf{v}_1 + \lambda_2^n \textbf{v}_2 +\ldots + \lambda_N^n \textbf{v}_N,
\]
and since $|\lambda_i|<1,\ i \neq 1$,
%with $x(n)= 1/N \textbf{v}_1\textbf{x}(0),$
the consensus is asymptotically reached by all the agents in the network,
i.e., 
$\lim_{n\to \infty} \textbf{x}(n) = \textbf{v}_1 =
(\textbf{w}_1^T\textbf{x}(0)/\textbf{w}_1^T\textbf{1})\textbf{1}.$
The asymptotic convergence implies that
the exact consensus value will not be achieved in a finite
number of iterations.
In practice, the consensus is said to be achieved when
$|x_i(n) - x_j(n)|<\texttt{tol}$ for all $i$ and $j$,
and a prefixed error tolerance $\texttt{tol}$.
The convergence speed of \eqref{Consensus_Laplacian}
depends on $\max(|\lambda_2|,|\lambda_N|).$
When the size of the network is large
or the number of links is small this value is usually close to one,
which means that the algorithm requires many iterations
before obtaining a good approximation of the final solution.

When the communication topology changes with the time,
$\mathcal{G}(n)=\{\mathcal{V},\mathcal{E}(n)\},\ $
eq. \eqref{Consensus_Laplacian} becomes
$\textbf{x}(n) = \textbf{A}(n) \textbf{x}(n-1),$
where the different weight matrices are defined according to their respective
underlying communication graphs.
If the different weight matrices satisfy Assumption \ref{RowStochastic},
and the sequence of matrices is not degenerated,
the algorithm is still proved to achieve consensus.
We refer the reader to~\cite{FB-JC-SM:09} for further information
about this case.

%% file: algorithm.tex
The distributed evaluation of polynomials provides
an easy way to speed up the consensus,
keeping the good properties found in standard methods.
The main idea consists in designing a distributed linear iteration
such that the execution of a fixed number of $n$ steps is
equivalent to the evaluation of some polynomial, $P_n(x),$ in the fixed
matrix $\textbf{A}$~\cite{EK-PF:09,SS-CNH:07b}.
The polynomial must satisfy that $P_n(1)= 1$ and
$|P_n(x)|<1$ if $|x|<1.$
In this way, successive evaluations of the polynomial in $\textbf{A}$
will lead to the consensus.
The choice of the polynomial
%as well as the design of the update rule
determine the convergence speed of the algorithm,
given by $\max_{\lambda_i} |P_n(\lambda_i)|,$
with $\lambda_i$ the eigenvalues of $\textbf{A}$.
%In the next sections we propose a distributed algorithm
%based on Chebyshev Polynomials that lead to consensus,
%improving the convergence rate of existing approaches.

Two reasons motivate the choice of Chebyshev polynomials
for the consensus problem:
\begin{itemize}
\item By using the recurrent definition \eqref{ChebyshevSimple},
instead of considering a polynomial of fixed degree
we can evaluate Chebyshev polynomials of higher and higher
degree as successive iterations of the algorithm are executed.
\item Chebyshev polynomials have the mini-max property~\cite{JCM-DCH:02}.
This property says that, among all the monic polynomials of degree $n$, %$P_n(x),$
the polynomial $2^{1-n}T_n(x)$ is the one that minimizes the uniform norm
on the interval $[-1,1]$.
This property is indeed quite convenient for our purposes.
If the matrix $\textbf{A}$ is unknown, using the Chebyshev polynomials
we are minimizing $\max_{\lambda\in[-1,1]} P_n(\lambda),$
therefore, getting high chances to obtain a good convergence rate.
\end{itemize}

However, the monic version of the Chebyshev polynomials
does not satisfy $2^{1-n}T_n(1)=1.$
In order to keep this property
we perform a linear transformation of $T_n(x)$,
using two real coefficients $\lambda_m$, $\lambda_M$,
with $1>\lambda_M > \lambda_m >-1,$
bringing the interval $[\lambda_m,\lambda_M]$ to $[-1,1].$
%for a given $n,$ by eq. \eqref{Property1},
%Chebyshev polynomials have $n+1$ points for which $|T_n(x)| = 1$
%in the interval $x\in[-1,1].$
%In the case of Chebyshev polynomials,
%for a given $n,$ by eq. \eqref{Property1}, in the interval $x\in[-1,1]$
%there are $n+1$ points for which $|T_n(x)| = 1.$
%This implies that $T_n(x)$ is not smaller than 1 for all $x\in (-1,1)$
%and we can not ensure convergence to the consensus
%if we use this polynomial in the process.
%We overcome this limitation
%using two real coefficients $\lambda_m$, $\lambda_M$,
%satisfying $1>\lambda_M > \lambda_m >-1,$
%using them to execute a linear transformation
%that brings $[\lambda_m,\lambda_M]$ to $[-1,1].$
In this way, we define the polynomial
%\begin{equation}
%\label{ChebyConsensus}
%p_n(x) = \frac{T_n(cx-d)}{T_n(c-d)},
%\end{equation}
%with
%\begin{equation}
%\label{Parameters_Consensus}
%c=\dfrac{2}{\lambda_M - \lambda_m}, \quad d=\dfrac{\lambda_M + \lambda_m}{\lambda_M - \lambda_m},
%\end{equation}
\begin{equation}
\label{ChebyConsensus}
P_n(x) = \frac{T_n(cx-d)}{T_n(c-d)},
\hbox{ with }
c=\dfrac{2}{\lambda_M - \lambda_m}, \quad d=\dfrac{\lambda_M + \lambda_m}{\lambda_M - \lambda_m},
\end{equation}
which, for all $n,$ has the following properties:
\begin{itemize}
\item
if $x\in[\lambda_m,\lambda_M]$, then $cx-d\in[-1,1]$
\item
$P_n(1)=1$ and $P_n(\lambda_M+\lambda_m-1)=(-1)^n$
%\item
%$\max_{x \in (\lambda_M+\lambda_m-1, 1)} |P_n(x)|<1$
\item
$|P_n(x)| <1$ for all $x \in (\lambda_M+\lambda_m-1, 1)$
and $|P_n(x)| \geq1$ otherwise.
%\item
%$|P_n(x)| >1$ for all $x \not\in [\lambda_M+\lambda_m-1, 1]$
\end{itemize}

The polynomial defined in \eqref{ChebyConsensus} satisfies the recurrence
\begin{equation}
\label{RecurrencePn}
%\begin{split}
P_{n}(x) = 2\dfrac{T_{n-1}(c-d)}{T_{n}(c-d)}(cx-d)P_{n-1}(x)-%\\
\dfrac{T_{n-2}(c-d)}{T_{n}(c-d)}P_{n-2}(x)
%\end{split}
\end{equation}
and the consensus rule  $\textbf{x}(n)= P_n(\textbf{A})\textbf{x}(0)$
is defined by
\begin{equation}
\label{ConsensusAlg_Cheby}
\begin{split}
\textbf{x}(1) &= P_1(\textbf{A})\textbf{x}(0) = \dfrac{1}{T_1(c-d)}(c \textbf{A}  - d \textbf{I})\textbf{x}(0),\\[10pt]
\textbf{x}(n) &= P_n(\textbf{A})\textbf{x}(0) =
%\frac{2(c \textbf{A} -d \textbf{I})T_{n-1}(c \textbf{A} -d \textbf{I})}{T_n(c-d)}\textbf{x}(0) -
%\frac{T_{n-2}(c \textbf{A} -d \textbf{I})}{T_n(c-d)}\textbf{x}(0) =\\[10pt]
\left (2\dfrac{T_{n-1}(c-d)}{T_{n}(c-d)}(c \textbf{A} -d \textbf{I})P_{n-1}(\textbf{A})
-\dfrac{T_{n-2}(c-d)}{T_{n}(c-d)} P_{n-2}(\textbf{A})\right )\textbf{x}(0) \\[10pt]
&=2\dfrac{T_{n-1}(c-d)}{T_{n}(c-d)}(c \textbf{A} -d \textbf{I})\textbf{x}(n-1)
-\dfrac{T_{n-2}(c-d)}{T_{n}(c-d)} \textbf{x}(n-2),\ n\geq2,
\end{split}
\end{equation}
with $\textbf{I}$ the identity matrix of dimension $N$.
Notice that this consensus rule is well designed to be executed in a distributed
fashion.
%This consensus rule allows a stable computation of successive iterations
%in a distributed way only by transmitting the current state to the neighbors,
%as in \eqref{Consensus_Laplacian}.
%The only additional information required in the algorithm
%are $\lambda_m$ and $\lambda_M.$

When the topology of the network changes,
the recurrent evaluation of Chebyshev polynomials \eqref{ConsensusAlg_Cheby}
can still be used.
%polynomial approaches can still be used,
%provided that they are defined by a distributed linear iteration.
The time-varying version of the algorithm is equivalent to \eqref{ConsensusAlg_Cheby}
replacing the constant weight matrix $\textbf{A}$ by the weight matrix at each step $\textbf{A}(n).$
Although this is no longer equivalent to the distributed evaluation of a Chebyshev polynomial,
a theoretical analysis about its convergence properties is still possible.
Algorithm \ref{ConsensusAlg} shows a possible implementation of the algorithm.
%because there is no a unique matrix $\textbf{A}$.
In the rest of the paper we analyze, both in theory and practice,
the main properties of this algorithm
for fixed and switching communication topologies.
\begin{algorithm}[!ht]
\small
\caption{Consensus algorithm using Chebyshev polynomials - agent $i$}
\label{ConsensusAlg}
\begin{algorithmic}[1]
\Require $x_i(0),$ MaxIt $\in \mathbb{N},\ \lambda_m,\ \lambda_M,$
%\Require $x_i(0),$ MaxIt $\in \mathbb{N},\ \lambda_m,\ \lambda_M,\ a_{ij},\ j \in \mathcal{N}_i$
%\begin{itemize}
%\item $\textbf{A}=[a_{ij}]$ satisfies assumption \ref{RowStochastic}
%\item $1>\lambda_M > \lambda_m >-1,$ equal for all $i \in \mathcal{V}$
%with $\lambda_N > \lambda_m +\lambda_M-1$
%\end{itemize}
%\Ensure $x_i^{(2)} \to \bar{x} = 1/N \textbf{1}^T \textbf{x}(0)$
\State -- \textit{Initialization}
\State $c=2/(\lambda_M - \lambda_m); \quad d=(\lambda_M + \lambda_m)/(\lambda_M - \lambda_m);$
%\State $a^{(0)}=1$;\quad $x_i^{(0)}=x_i(0);$ $a^{(1)}=c-d$;
\State $T(0)=1;\quad T(1)=c-d$;
\State -- \textit{First Communication Round}
%\[x_i^{(1)} = \frac{1}{a^{(1)}}(c\sum_{j\in\mathcal{N}_i(n)} a_{ij}x_j^{(0)}+(c\ a_{ii}-d) x_i^{(0)});\]
\[x_i(1) = \frac{1}{T(1)}(c\sum_{j\in\mathcal{N}_i(n)} a_{ij}x_j(0)+(c\ a_{ii}-d) x_i(0));\]
%\Repeat
%\State $a^{(2)}= 2 (c-d) a^{(1)}- a^{(0)};$
%\State -- \textit{Communication Between Neighbors}
%\[x_i^{(2)} = 2 \dfrac{a^{(1)}}{a^{(2)}}(c\sum_{j\in\mathcal{N}_i(n)} a_{ij}x_j^{(1)}+(c\ a_{ii}-d) x_i^{(1)})- \dfrac{a^{(0)}}{a^{(2)}} x_i^{(0)};\]
%\State -- \textit{Update Parameters}
%\State $a^{(0)}=a^{(1)};\quad x_i^{(0)}=x_i^{(1)};\quad$
%$a^{(1)}=a^{(2)};\quad x_i^{(1)}=x_i^{(2)};$
%\Until{MaxIt}
%\EndFor
\For{$n=2,\ldots,$MaxIt}
%\State $a^{(n)}= 2 (c-d) a^{(n-1)}- a^{(n-2)};$
\State $T(n)= 2 (c-d) T(n-1)- T(n-2);$
\State -- \textit{Communication Between Neighbors}
%\[x_i^{(n)} = 2 \dfrac{a^{(n-1)}}{a^{(n)}}(c\sum_{j\in\mathcal{N}_i(n)} a_{ij}x_j^{(n-1)}+(c\ a_{ii}-d) x_i^{(n-1)})- \dfrac{a^{(n-2)}}{a^{(n)}} x_i^{(n-2)};\]
\[x_i(n) = 2 \dfrac{T(n-1)}{T(n)}(c\sum_{j\in\mathcal{N}_i(n)} a_{ij}x_j(n-1)+(c\ a_{ii}-d) x_i(n-1))- \dfrac{T(n-2)}{T(n)} x_i(n-2);\]
%\State -- \textit{Update Parameters}
%\State $a^{(0)}=a^{(1)};\quad x_i^{(0)}=x_i^{(1)};\quad$
%$a^{(1)}=a^{(2)};\quad x_i^{(1)}=x_i^{(2)};$
\EndFor
\end{algorithmic}
\end{algorithm}
%%%%%%%%%%%%%%%%%%%%%%%%%%%%%%%%%%%%%%%%%%%%%%%%%%%%%%%%%%%%%%
%END ALGORITHM
%%%%%%%%%%%%%%%%%%%%%%%%%%%%%%%%%%%%%%%%%%%%%%%%%%%%%%%%%%%%%% 

%% file: polynomials2.tex
In this section we analyze the main properties of
the proposed algorithm when the network topology is fixed.
In particular we first study the convergence conditions of the algorithm.
Next, we find the parameters that maximize the convergence speed.
Finally,
we give bounds on the selection of these parameters
to satisfy that our algorithm achieves the consensus faster
than \eqref{Consensus_Laplacian}.

%%%%%%%%%%%%%%%%%%%%%%%%%%%%%%%%%%%%%%%%%%%%%%%%%%%%%%%%%%%%%%%%%%%%%%%%%%%%%%
%THEOREM
%%%%%%%%%%%%%%%%%%%%%%%%%%%%%%%%%%%%%%%%%%%%%%%%%%%%%%%%%%%%%%%%%%%%%%%%%%%%%%
\begin{theorem}[Convergence of the algorithm]
\label{ConvergenceTh}
Let $\textbf{A}$
be diagonalizable, fulfilling Assumption \ref{RowStochastic},
and parameters $\lambda_m$ and $\lambda_M$ such
that $1 > \lambda_M > \lambda_m > -1$.
If the minimum real eigenvalue of $\textbf{A}$
satisfies $\lambda_N > \lambda_m + \lambda_M -1$ and the complex eigenvalues,
$\lambda_z,$ of $\textbf{A}$ satisfy $|\tau(c\lambda_z-d)|>\tau(c-d),$
%\begin{equation}
%\min \{|c\lambda-d - \sqrt{(c\lambda-d)^2 -1}|, |c\lambda-d + \sqrt{(c\lambda-d)^2 -1}|\}
%< c-d - \sqrt{(c-d)^2-1} \},
%\end{equation}
%\begin{equation}
%\begin{split}
%&\min \{|c\lambda-d - \sqrt{(c\lambda-d)^2 -1}|, |c\lambda-d + \sqrt{(c\lambda-d)^2 -1}|\}\\
%&< c-d - \sqrt{(c-d)^2-1} \},
%\end{split}
%\end{equation}
then the recurrence in eq. \eqref{ConsensusAlg_Cheby} converges
to the consensus state,
$\lim_{n\to \infty} \textbf{x}(n) = \textbf{w}_1^T\textbf{x}(0)\textbf{1}/\textbf{w}_1^T\textbf{1}.$
Besides,
the convergence rate is given by
\begin{equation}
\label{convergenceRate1}
%\Vert \textbf{x}(n) - \textbf{v}_1 \Vert_2 \le
\displaystyle\max_{\lambda_i \ne 1} \dfrac{\left|T_n(c\lambda_i -d)\right|}{T_n(c-d)}
%\Vert \textbf{x}(0) - \textbf{v}_1 \Vert_2.
\end{equation}
\end{theorem}
%%%%%%%%%%%%%%%%%%%%%%%%%%%%%%%%%%%%%%%%%%%%%%%%%%%%%%%%%%%%%%%%%%%%%%%%%%%%%%
%PROOF
%%%%%%%%%%%%%%%%%%%%%%%%%%%%%%%%%%%%%%%%%%%%%%%%%%%%%%%%%%%%%%%%%%%%%%%%%%%%%%
\begin{proof}
See the Appendix.
\end{proof}
%%%%%%%%%%%%%%%%%%%%%%%%%%%%%%%%%%%%%%%%%%%%%%%%%%%%%%%%%%%%%%%%%%%%%%%%%%%%%%
%END PROOF
%%%%%%%%%%%%%%%%%%%%%%%%%%%%%%%%%%%%%%%%%%%%%%%%%%%%%%%%%%%%%%%%%%%%%%%%%%%%%%

Note that the conditions in Theorem \ref{ConvergenceTh} are easy to
fulfill without the necessity of knowing the eigenvalues of the matrix $\textbf{A}.$
For the real eigenvalues, any symmetric selection of
the parameters, i.e., $-\lambda_m = \lambda_M,\ 0<\lambda_M<1,$
satisfies the condition in Theorem \ref{ConvergenceTh}.
The condition on the complex eigenvalues has some
geometrical meaning~\cite{JCM-DCH:02}.
Imposing that $|\tau(c\lambda_z-d)|>\tau(c-d)$
is equivalent to require that $\lambda_z$ is
inside an ellipse in the complex plane centered at $(d/c, 0),$
or equivalently $((\lambda_M+\lambda_m)/2, 0),$
and with semi-axis $e_1 = (c-d)/c$ and $e_2 = (\sqrt{(c-d)^2-1})/c$
(see Fig \ref{EllipseConvergence}).
In practice, any parameters that ensure convergence for the real
eigenvalues also ensure convergence for the complex ones.
We have observed that if $\textbf{A}$ is defined
using well known distributed methods~\cite{LX-SB:04},
the complex eigenvalues,
when there are any of them, have always a very small modulus.
For that reason, in the rest of the section
we will assume that the matrix $\textbf{A}$ has only real eigenvalues.
\begin{figure}[!ht]
  \centering
 \includegraphics[width=0.6\linewidth]{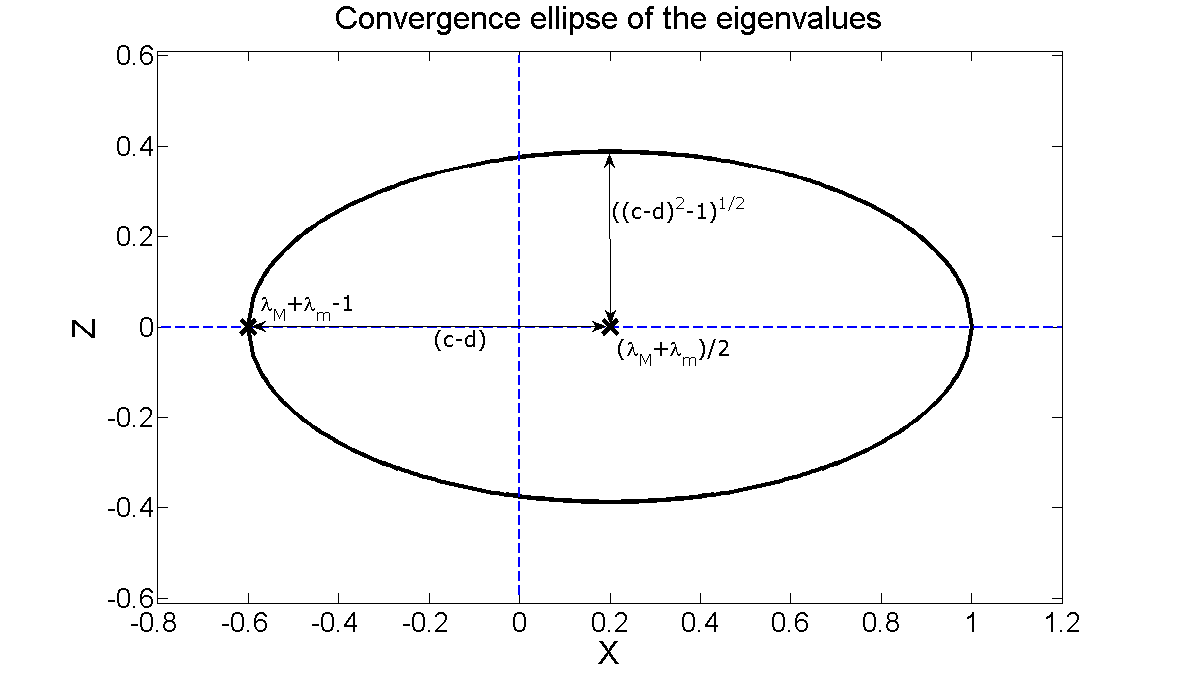}
  \caption{
  Ellipse where all the eigenvalues must be contained in order to achieve the consensus.
  In this particular example we have chosen $\lambda_M = 0.9$ and $\lambda_m = -0.5.$
  Note that when the imaginary part of the eigenvalues is zero convergence is achieved
  if $\lambda_M+\lambda_m-1>\lambda>1$ as stated in Theorem \ref{ConvergenceTh}.
  }
  \label{EllipseConvergence}
\end{figure}

Next, we are interested in knowing the optimal
selection of $\lambda_m$ and $\lambda_M$ to maximize the convergence speed.
From Theorem \ref{ConvergenceTh}
we know that the convergence rate is given by the factor
%\begin{equation*}
%\max _{\lambda_i\neq1} \dfrac{T_n(c\lambda_i-d)}{T_n(c-d)}.
%\end{equation*}
%Let us suppose that
%$[\lambda_N, \lambda_2] \not\subseteq [\lambda_n, \lambda_M]$.
%Then
\begin{equation}
\small
\max_{\lambda_i\neq1} \dfrac{|T_n(c\lambda_i-d)|}{T_n(c-d)}=
\max\left\{ \dfrac{|T_n(c\lambda_N-d)|}{T_n(c-d)},
 \dfrac{|T_n(c\lambda_2-d)|}{T_n(c-d)} \right\}.
\end{equation}
%and $\max\{ |c\lambda_N -d|, |c \lambda_2 -d|\} >1$.
%We are interested in the values of $\lambda_m$, $\lambda_M$
%that make minimum this value (for $n \to \infty$).

If the conditions in Theorem \ref{ConvergenceTh}
are satisfied, for any $\lambda,$
a simple calculation using eq. \eqref{direct_Expression} leads to
\begin{equation}
\label{TnFunctionTau}
\dfrac{|T_n(c\lambda-d)|}{T_n(c-d)}=
\left( \dfrac{\tau(c-d)}{|\tau(c\lambda - d)|} \right)^n
\dfrac{1+ \tau(c\lambda-d)^{2n}}{1+ \tau(c-d)^{2n}}.
\end{equation}
It is clear that when $n \to \infty$, the second fraction in the right side
of \eqref{TnFunctionTau} goes to 1.
Therefore, the convergence rate is determined by
\begin{equation}
\max\left\{ \dfrac{\tau(c-d)}{|\tau(c\lambda_N - d)|},
 \dfrac{\tau(c-d)}{|\tau(c\lambda_2 - d)|} \right\}
\end{equation}
If $[\lambda_N, \lambda_2] \subseteq [\lambda_m, \lambda_M]$, then
$\max_{\lambda_i} |T_n(c\lambda_i-d)| \le 1$ and therefore we can define
the convergence factor as
\begin{equation}
\nu(c,d)=
\left\{
\begin{split}
&\tau(c-d), &\hbox{ if } [\lambda_N, \lambda_2] \subseteq [\lambda_m, \lambda_M]\\[7pt]
&\max\left\{ \dfrac{\tau(c-d)}{|\tau(c\lambda_N - d)|},
\dfrac{\tau(c-d)}{|\tau(c\lambda_2 - d)|} \right\},  &\hbox{ otherwise.}
\end{split}
\right.
\end{equation}

The optimum values of $\lambda_m$ and  $\lambda_M$
will be those that lead to the minimum value of $\nu(c,d)$.
In~\cite{EM-JIM-CS:11} %(Proposition 3.1)
it was proved that among the values of the parameters
satisfying $[\lambda_N, \lambda_2] \subseteq [\lambda_n, \lambda_M],$
the ones that yield the minimum convergence factor are precisely
$\lambda_m=\lambda_N$ and $\lambda_M=\lambda_2$.
Let us see that they are also the optimum parameters in the case
$[\lambda_N, \lambda_2] \not\subseteq [\lambda_n, \lambda_M].$

%%%%%%%%%%%%%%%%%%%%%%%%%%%%%%%%%%%%%%%%%%%%%%%%%%%%%%%%%%%%%%%%%%%%%%%%%%%%%%
%THEOREM
%%%%%%%%%%%%%%%%%%%%%%%%%%%%%%%%%%%%%%%%%%%%%%%%%%%%%%%%%%%%%%%%%%%%%%%%%%%%%%
\begin{theorem}[Optimal parameters]
\label{Th_optimalParameters}
The convergence rate $\nu(c,d)$ attains its minimum value for
the parameters $c,\ d$ such that $\lambda_M=\lambda_2$ and $\lambda_m=\lambda_N$
\end{theorem}
%%%%%%%%%%%%%%%%%%%%%%%%%%%%%%%%%%%%%%%%%%%%%%%%%%%%%%%%%%%%%%%%%%%%%%%%%%%%%%
%PROOF
%%%%%%%%%%%%%%%%%%%%%%%%%%%%%%%%%%%%%%%%%%%%%%%%%%%%%%%%%%%%%%%%%%%%%%%%%%%%%%
\begin{proof}
See the Appendix.
\end{proof}
%%%%%%%%%%%%%%%%%%%%%%%%%%%%%%%%%%%%%%%%%%%%%%%%%%%%%%%%%%%%%%%%%%%%%%%%%%%%%%
%END PROOF
%%%%%%%%%%%%%%%%%%%%%%%%%%%%%%%%%%%%%%%%%%%%%%%%%%%%%%%%%%%%%%%%%%%%%%%%%%%%%%
This implies that in order to achieve the maximum convergence speed,
some knowledge about the network is required.
However, even if the network topology is unknown,
it is important to study when the algorithm converges
in a faster way than \eqref{Consensus_Laplacian}.
Since the symmetric assignation of the parameters, $\lambda_M=-\lambda_m,$
always ensures convergence,
in the last result of this section we provide bounds for this particular case
that also converge faster than \eqref{Consensus_Laplacian}.
%%%%%%%%%%%%%%%%%%%%%%%%%%%%%%%%%%%%%%%%%%%%%%%%%%%%%%%%%%%%%%%%%%%%%%%%%%%%%%
%THEOREM
%%%%%%%%%%%%%%%%%%%%%%%%%%%%%%%%%%%%%%%%%%%%%%%%%%%%%%%%%%%%%%%%%%%%%%%%%%%%%%
\begin{theorem}[Faster convergence than $\textbf{A}^n$]
\label{Th_convFaster}
For any matrix $\textbf{A}$ satisfying
Assumption \ref{RowStochastic},
let $\lambda = \max(|\lambda_2|,|\lambda_N|)$
be the convergence rate in \eqref{Consensus_Laplacian}.
For any
\begin{equation}
\label{conditions_onM}
0 < \lambda_M < \frac{2\lambda}{\lambda^2+1},\hbox{ and } \lambda_m = -\lambda_M,
\end{equation}
$P_n(\lambda)$ goes to zero
faster than $\lambda^n$ when $n$ goes to infinity.
Therefore the algorithm in eq. \eqref{ConsensusAlg_Cheby} converges
to the consensus
faster than the one in eq. \eqref{Consensus_Laplacian}.
\end{theorem}
\vskip 0.2cm
%%%%%%%%%%%%%%%%%%%%%%%%%%%%%%%%%%%%%%%%%%%%%%%%%%%%%%%%%%%%%%%%%%%%%%%%%%%%%%
%PROOF
%%%%%%%%%%%%%%%%%%%%%%%%%%%%%%%%%%%%%%%%%%%%%%%%%%%%%%%%%%%%%%%%%%%%%%%%%%%%%%
\begin{proof}
See~\cite{EM-JIM-CS:11}.
\end{proof}
%%%%%%%%%%%%%%%%%%%%%%%%%%%%%%%%%%%%%%%%%%%%%%%%%%%%%%%%%%%%%%%%%%%%%%%%%%%%%%
%END PROOF
%%%%%%%%%%%%%%%%%%%%%%%%%%%%%%%%%%%%%%%%%%%%%%%%%%%%%%%%%%%%%%%%%%%%%%%%%%%%%%
\begin{remark}
The above result shows that there always exist parameters that make
the proposed algorithm faster than \eqref{Consensus_Laplacian}.
Therefore, if the algorithm is executed using the optimal parameters,
it will also converge to the average faster than \eqref{Consensus_Laplacian}.
\end{remark}

%\begin{comment}
\begin{figure}[!ht]
  \centering
 \includegraphics[width=0.49\linewidth]{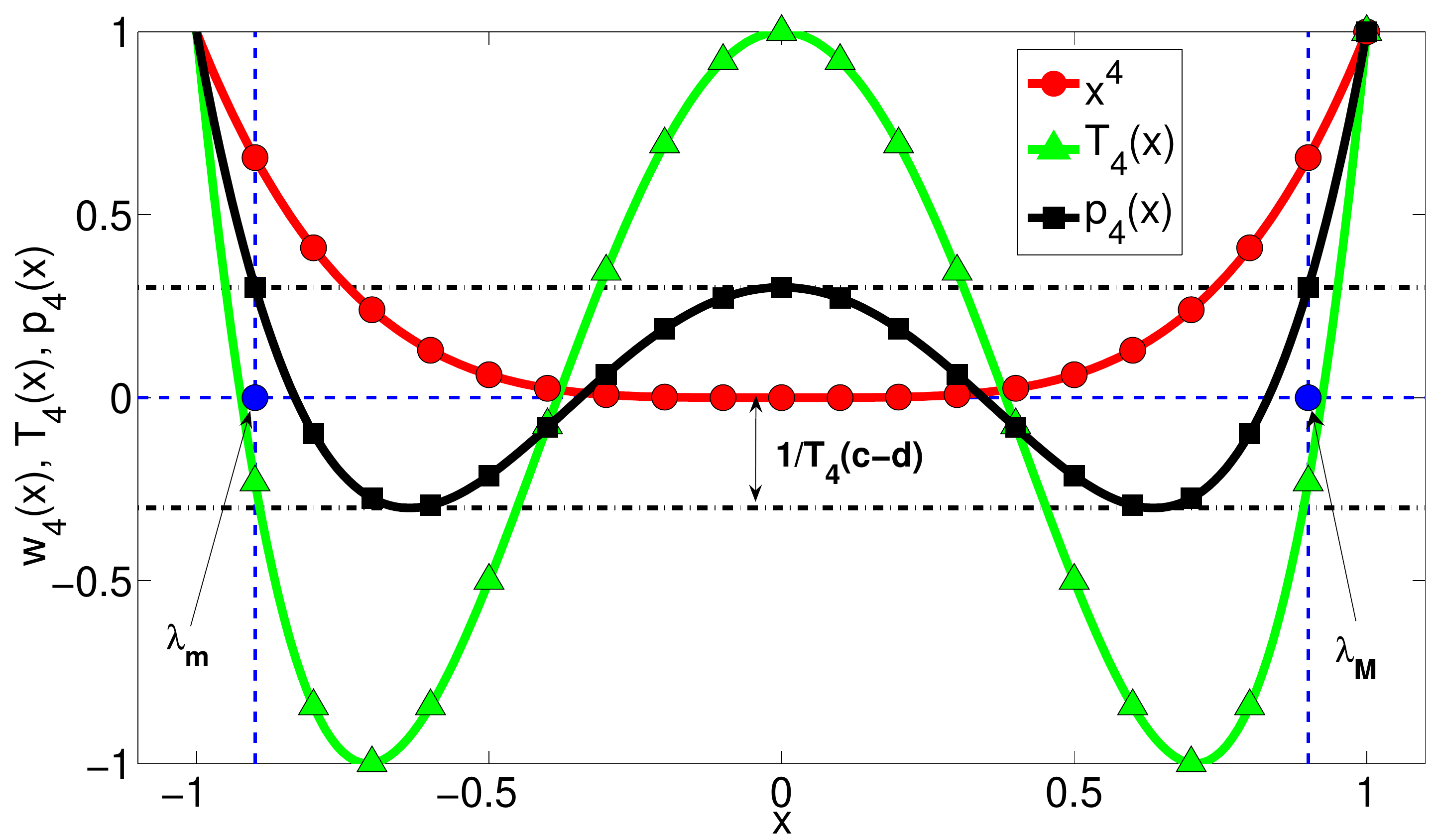}
  \caption{Plot of the polynomials $x^n,\ T_n(x)$ and $P_n(x)$.
  In the figure $n=4, \lambda_m = -0.95$ and $\lambda_M = 0.95.$
  }
  \label{ChebPoltransformed}
\end{figure}
Finally, a graphical comparison of $x^n,\ T_n(x)$
and $P_n(x)$ is depicted in Fig. \ref{ChebPoltransformed}
for $n=4,$ in the interval $[-1,1].$
%We observe
Note that $T_n(x)$ cannot be used in the consensus process
because at some points it would not reduce the error.
On the other hand, as we have shown along the section,
$P_n(x)$ satisfies the conditions required to achieve consensus.
Also notice that $P_n(x)$ has closer
values to zero than $x^n$ in points close to $-1$ and $1,$
which supports the theory that the error associated to eigenvalues
in that regions will be reduced faster.
%\end{comment}

%% file: switching.tex
%cannot be applied because there is no a unique matrix $\textbf{A}$.
%However, the recurrent evaluation of Chebyshev polynomials \eqref{ConsensusAlg_Cheby}
%is suitable for switching weight matrices.
%In this section we analyze the behavior of \eqref{ConsensusAlg_Cheby}
%when the communication topology changes.
We are interested now in the study of the recursive evaluation of \eqref{ConsensusAlg_Cheby}
when the topology of the network, and therefore the matrix $\textbf{A}$, changes
at different iterations.
Given initial conditions $\textbf{x}(0),$
the distributed recurrence now looks:
\begin{equation}
\label{ChebyshevRecurrenceSwitch}
\begin{split}
\textbf{x}(1) &= \dfrac{1}{T_1(c-d)}(c \textbf{A}(1)  - d \textbf{I})\textbf{x}(0),\\[10pt]
\textbf{x}(n) &=2\dfrac{T_{n-1}(c-d)}{T_{n}(c-d)}(c \textbf{A}(n) -d \textbf{I})\textbf{x}(n-1)
-\dfrac{T_{n-2}(c-d)}{T_{n}(c-d)} \textbf{x}(n-2),\ n\geq2.
\end{split}
\end{equation}
Note that this recurrence is suitable for switching weight matrices.
However, the evaluation of the recurrence is no
longer equivalent to $P_n(\textbf{A})\textbf{x}(0),$
for some matrix $\textbf{A}$.
This means that we are not exactly evaluating 
the transformed Chebyshev polynomials in the eigenvalues of some matrix anymore.
Nevertheless, a theoretical analysis is still possible.

For this analysis, the matrices $\textbf{A}(n)$ now require the following assumption.
\begin{assumption}[Non-Degenerate Stochastic Weights]
\label{TimeVaryingStochasticMatrices}
The matrices $\textbf{A}(n)$
are row stochastic, symmetric, non-degenerate and
compatible with the underlying graphs, $\mathcal{G}(n),$ for all $n$,
i.e., they are such that
$\textbf{A}(n) \textbf{1} = \textbf{1},$
$a_{ii}(n) > \epsilon$
and $a_{ij}(n)\in\{0\}\cup[\epsilon,1)$
with $0<\epsilon<1$ some fixed constant.
\end{assumption}

\begin{comment}
The iteration in \eqref{ChebyshevRecurrenceSwitch} is equivalent to
\begin{equation}
\label{ChebyshevRecurrenceSwitch3}
\textbf{x}(n) = \frac{\textbf{u}(n)}{{T_{n}(c-d)}},
\end{equation}
with
\begin{equation}
\label{ChebyshevRecurrenceSwitch2}
\begin{split}
\textbf{u}(0) &=\textbf{x}(0),\
\textbf{u}(1) =(c \textbf{A}(1)  - d \textbf{I})\textbf{x}(0),\\
\textbf{u}(n) &= 2(c \textbf{A}(n) -d \textbf{I})\textbf{u}(n-1)
-\textbf{u}(n-2),\ n\geq2,
\end{split}
\end{equation}
which will be more convenient for our analysis.
Let us note that by definition in \eqref{ChebyshevRecurrenceSwitch3},
$T_n(c-d)$ does not change, even if the network does.
\end{comment}

Recalling the analysis done in the previous section,
the evaluation of $P_n(\textbf{A})$ was separated into the evaluation
of its eigenvalues and eigenvectors, 
$P_n(\lambda_i)\textbf{v}_i=T_n(c\lambda_i-d)/T_n(c-d)\textbf{v}_i$.
In the switching case we must take into account that both $\lambda_i$ and $\textbf{v}_i$
change at each iteration. Moreover, since the eigenvectors
of different matrices are related we must also consider these relations.
For the moment, as a first simplification of the problem,
let us forget about the changes in $\textbf{v}_i$ and the parameters $c$ and $d$
and let us study the scalar evaluation of the Chebyshev recurrence \eqref{ChebyshevSimple}
with different $\lambda_i$ at each iteration.
%In order to study the evolution of \eqref{ChebyshevRecurrenceSwitch}, %\eqref{ChebyshevRecurrenceSwitch2},
%we analyze the scalar evaluation of the Chebyshev recurrence \eqref{ChebyshevSimple}
%with different $\lambda$ at each iteration as a first simplification of the problem.
That is,
\begin{equation}
T_0(\Lambda) = 1,\
T_1(\Lambda) = \lambda(1),\
T_{n}(\Lambda) = 2 \lambda(n) T_{n-1}(\Lambda) - T_{n-2}(\Lambda),
\end{equation}
where $\Lambda = \{\lambda(n)\},\ n\in\mathbb{N}$ is a succession of real numbers.
Specifically, we are interested in the behavior of $|T_{n}(\Lambda)|.$
\begin{proposition}
\label{PropositionBoundUn}
Suppose there exists values $\lambda_{\min}$ and $\lambda_{\max}$ such that
$\lambda(n) \in [\lambda_{\min},\lambda_{\max}],\ \forall n\in\mathbb{N},$
$\lambda_{\min} < 0 < \lambda_{\max}$
and $|\lambda_{\min}|\leq \lambda_{\max}$. Then
\begin{equation}
\label{BoundRecurrenceUn}
|T_n(\Lambda)| \leq |T_n(\Lambda^*)|
\end{equation}
where $\Lambda^* = \{\lambda^*(n)\}$ is a succession defined by
\begin{equation}
\label{BestSuccession}
\lambda^*(n) =
\begin{cases}
\lambda_{\max} &\hbox{ if } n \hbox{ odd,}\\
\lambda_{\min} &\hbox{ if } n \hbox{ even,}\\
\end{cases}
\end{equation}
\end{proposition}
\begin{proof}
For abbreviation,
in the proof we will denote the sign of $T_n(\Lambda)$ by $s(T_n).$
%We show the result by induction.
%Now let us suppose that \eqref{BoundRecurrenceUn} is also true for $n-1$ and $n-2.$
%Let us first consider the case with $n$ even.

Let us note that, if $s(T_{n-1})=s(T_{n-2})$, by choosing $\lambda(n) < 0,$ then
\begin{equation}
\label{equalitySigns}
|T_n(\Lambda)| = |2 \lambda(n) T_{n-1}(\Lambda) - T_{n-2}(\Lambda)| = |2 \lambda(n) T_{n-1}(\Lambda)| + |T_{n-2}(\Lambda)|,
\end{equation}
independently of $n.$
The choice of $\lambda(n) > 0$ when $s(T_{n-1})=s(T_{n-2})$
implies that
\begin{equation}
\label{equalitySigns2}
|T_n(\Lambda)| = |2 \lambda(n) T_{n-1}(\Lambda) - T_{n-2}(\Lambda)| < |2 \lambda(n) T_{n-1}(\Lambda)| + |T_{n-2}(\Lambda)|.
\end{equation}
Taking these two facts into account we can see that
\begin{equation}
\label{minimumDifSign}
s(T_{n-1})=s(T_{n-2}) \Rightarrow \arg \max_{\lambda(n)} |T_n(\Lambda)| = \lambda_{\min}.
\end{equation}
Besides, in this situation, choosing $\lambda(n)<0$ yields $s(T_n)\neq s(T_{n-1})$.

Now, if $s(T_{n-1})\neq s(T_{n-2})$ and $\lambda(n) > 0,$ then
eq. \eqref{equalitySigns} is again true.
On the other hand, choosing $\lambda(n) < 0$ in this situation implies \eqref{equalitySigns2}.
Thus,
\begin{equation}
\label{maximumDifSign}
s(T_{n-1})\neq s(T_{n-2}) \Rightarrow \arg \max_{\lambda(n)} |T_n(\Lambda)| = \lambda_{\max}.
\end{equation}
Also, if $s(T_{n-1})\neq s(T_{n-2})$ and $\lambda(n)>0,$ then $s(T_{n}) = s(T_{n-1}).$

Finally, noting that
inequality \eqref{BoundRecurrenceUn} holds for $n=0$ and $1,$
and $s(T_{0}(\Lambda^*))=s(T_{1}(\Lambda^*)),$ then using \eqref{minimumDifSign} and
\eqref{maximumDifSign} the succession \eqref{BestSuccession} is obtained
and the result is proved.
\end{proof}
\begin{corollary}
If $|\lambda_{\min}| > \lambda_{\max}$ then the bound in eq. \eqref{BoundRecurrenceUn}
is true taking $\Lambda^* = \{\lambda^*(n)\}$ with
\begin{equation}
\lambda^*(n) =
\begin{cases}
\lambda_{\max} &\hbox{ if } n \hbox{ even,}\\
\lambda_{\min} &\hbox{ if } n \hbox{ odd,}\\
\end{cases}
\end{equation}
\end{corollary}
\vskip 0.5pt
%\begin{remark}
%Let us note that
%$|U_n(x(n))| \leq 2 |x(n)| U_{n-1}(x(n-1)) + U_{n-2}(x(n-2)).$
%\end{remark}

The previous proposition reveals that the Chebyshev recurrence
evaluated in a succession of different real numbers
does not keep the behavior shown when it is evaluated with a constant value.
The next Lemma provides a bound for the direct expression of this behavior.
\begin{lemma}
\label{LemmaSwitch}
Let us suppose that the conditions of Proposition \ref{PropositionBoundUn}
are true. Then
\begin{equation}
\label{directRecuSwitch}
|T_n(\Lambda^*)| \leq \kappa_1(\lambda_{\max})^n, \hbox{ where } \kappa_1(\lambda_{\max}) = \lambda_{\max} + \sqrt{\lambda_{\max}^2+1}
\end{equation}
\end{lemma}
\begin{proof}
Let us define the recurrence
\begin{equation}
\label{recuUn}
T^*_{0}(\lambda)=1,\ T^*_{1}(\lambda)=\lambda,\ T_n^*(\lambda)=2\lambda T^*_{n-1}(\lambda)+T^*_{n-2}(\lambda),
\end{equation}
which satisfies that
\begin{equation}
|T_n(\Lambda^*)| \leq T_n^*(\lambda_{\max}). %= 2x_{\max}T^*_{n-1}(x_{\max})+T^*_{n-2}(x_{\max}),
\end{equation}
%with $T^*_{1}(x_{\max})=x_{\max}$ and $T^*_{0}(x_{\max})=1.$

According to recurrence \eqref{recuUn}, the succession
$\{T_n^*(\lambda_{\max}),\ n=0, 1, \ldots\}$ satisfies the homogeneous difference equation
$T_n^*(\lambda_{\max}) - 2\lambda_{\max}T^*_{n-1}(\lambda_{\max})-T^*_{n-2}(\lambda_{\max}) = 0$.
%whose characteristic polynomial is $t^2 - 2x_{\max}t -1$.
By the theory of difference equations~\cite{RPA:92},
the solution to this equation is determined by the roots
$\kappa_1$ and $\kappa_2$ of the characteristic polynomial.
In this case
\begin{equation}
\kappa_1(\lambda_{\max}) = \lambda_{\max}+\sqrt{\lambda^2_{\max}+1} > 1,
\hbox{  and }
\kappa_2(\lambda_{\max}) = \lambda_{\max}-\sqrt{\lambda^2_{\max}+1}=-1/\kappa_1(\lambda_{\max}).
\end{equation}
Since $\kappa_1(\lambda_{\max})\ne \kappa_2(\lambda_{\max}),$ the direct expression of $T_n^*(\lambda_{\max})$ is
\begin{equation}
T_n^*(\lambda_{\max})= A \kappa_1(\lambda_{\max})^n + B \kappa_2(\lambda_{\max})^n
\end{equation}
where $A$ and $B$ depend on the initial conditions $T_0^*(\lambda_{\max})$ and $T_1^*(\lambda_{\max})$.
In our case $A=B=1/2$ and
\begin{equation}
|T_n(\Lambda^*)| \leq T_n^*(\lambda_{\max})= \dfrac{1}{2}(\kappa_1(\lambda_{\max})^n +
(-1/\kappa_1(\lambda_{\max}))^{n}) \leq \kappa_1(\lambda_{\max})^n.
\end{equation}
\end{proof}

This direct expression \eqref{directRecuSwitch} 
will be helpful in the development of the convergence
analysis dealing with changing matrices and the parameters $c$ and $d$.
We provide now the main result, 
showing the convergence of the algorithm for the switching case.
\begin{theorem}
\label{ConvergeSwitchingTopTh}
Allow the communication graph, $\mathcal{G}(n),$
to arbitrarily change in such a way that it is connected
for all $n$,
with the weight matrices, $\textbf{A}(n),$ designed according to
Assumption \ref{TimeVaryingStochasticMatrices}.
Let us denote $\lambda_i(n),\ i=1,\ldots,N,$ the eigenvalues of $\textbf{A}(n)$ and
\begin{equation}
\lambda_{\max} = \max_{n} \max_{i=2,\ldots,N} \lambda_i(n), \hbox{ and }
\lambda_{\min} = \min_{n} \min_{i=2,\ldots,N} \lambda_i(n).
\end{equation}
Given fixed parameters $c$ and $d$,
a sufficient condition to guarantee convergence to
consensus of iteration \eqref{ChebyshevRecurrenceSwitch} is
\begin{equation}
\label{convergenceConditionSwitching}
\kappa_1(\max \{ |c\lambda_{\max}-d|, |c\lambda_{\min} - d| \}) \tau(c-d) < 1.
\end{equation}
%with $x_{\max} = \max \{ |c\lambda_{\max}-d|, |c\lambda_{\min} - d| \}.$
\end{theorem}
\begin{proof}
See the Appendix.
\end{proof}
%In the rest of the section we discuss in detail the
%meaning of the theorem and its implications.
%We start by giving more specific values of $\lambda_M$ and $\lambda_m,$
%and therefore on $c$ and $d$,
%that satisfy the condition in the theorem.
The next corollaries give more specific values of $\lambda_M$ and $\lambda_m,$
and therefore on $c$ and $d$,
that satisfy the condition in the theorem
to achieve convergence.

\begin{comment}
Let us fist analyze the conditions on the parameters for
a symmetric assignation, $-\lambda_m=\lambda_M=\lambda,$
because it is a simple case that we have seen that
with a fixed topology always ensures convergence.
Recall that with this assignation $c=1/\lambda$ and $d=0$.
Assuming $|c\lambda_{\max}-d|>|c\lambda_{\min}-d|$,
substituting $\kappa$ and $\tau$ by their values
in eq. \eqref{convergenceConditionSwitching}
and doing some simplifications we get
\begin{equation}
\label{parametersCond1}
\lambda^2 < (1-\lambda_{\max}^2).
\end{equation}
\end{comment}
\begin{corollary}
\label{CorollaryParam1}
Assume $|c\lambda_{\max}-d|>|c\lambda_{\min}-d|$
and a symmetric assignation, $-\lambda_m=\lambda_M=\lambda,$
of the parameters.
Then if
\begin{equation}
\label{parametersCond1}
\lambda^2 < (1-\lambda_{\max}^2),
\end{equation}
the algorithm converges.
\end{corollary}
\begin{proof}
Recall that with this assignation $c=1/\lambda$ and $d=0$.
Substituting $\kappa_1$ and $\tau$ by their values
in eq. \eqref{convergenceConditionSwitching}
and doing some simplifications eq. \eqref{parametersCond1} is obtained.
\end{proof}

\begin{comment}
If we prefer to give non-symmetric values for the parameters
we can impose a constraint on $\lambda_M$ and $\lambda_m$ such
that $x_{\max}$ is as small as possible.
If we know the values of $\lambda_{\max}$ and $\lambda_{\min},$
or some approximations,
the choice of $\lambda_m$ and $\lambda_M$
should be done in such a way that
\begin{equation}
|c\lambda_{\min}-d| = |c\lambda_{\max}-d|.
\end{equation}
In this way we are reducing the value of $x_{\max}$.
Clearing the equation yields
\begin{equation}
\lambda_M + \lambda_m = \lambda_{\max}+\lambda_{\min}.
\end{equation}
Doing again some, rather tedious, calculations
in eq. \eqref{convergenceConditionSwitching}
the second condition is obtained
%\begin{equation}
%\begin{split}
%(c-d)^2 - (x_{\max})^2 &> 1\\
%c(c-d + c\lambda-d) &> \frac{1}{(1-\lambda)}
%c^2(1+\lambda-(\lambda_M+\lambda_m)) &> \frac{1}{(1-\lambda)}
%\end{split}
%\end{equation}
\begin{equation}
\label{parametersCond2}
(\lambda_M-\lambda_m)^2 < 4(1-\lambda_{\max})(1-\lambda_{\min})
\end{equation}
\end{comment}

If we prefer to assign non-symmetric values to the parameters,
the following corollary provides a possible assignation that
satisfies Theorem \ref{ConvergeSwitchingTopTh}.
\begin{corollary}
\label{CorollaryParam2}
Assume now that the values of $\lambda_{\max}$ and $\lambda_{\min},$
or some bounds, are known.
If $\lambda_M$ and $\lambda_m$ satisfy that
\begin{equation}
\label{parametersCond3}
\lambda_M + \lambda_m = \lambda_{\max}+\lambda_{\min},
\end{equation}
and
\begin{equation}
\label{parametersCond2}
\lambda_M-\lambda_m < \sqrt{4(1-\lambda_{\max})(1-\lambda_{\min})},
\end{equation}
then the algorithm achieves the consensus.
\end{corollary}
\begin{proof}
If we know the values of $\lambda_{\max}$ and $\lambda_{\min},$
the choice of $\lambda_m$ and $\lambda_M$
can be done in such a way that
\begin{equation}
\label{parametersCond5}
|c\lambda_{\min}-d| = |c\lambda_{\max}-d|.
\end{equation}
With this assignation we are minimizing the value of $\max \{ |c\lambda_{\max}-d|, |c\lambda_{\min} - d| \}$
and therefore, the convergence condition is easier to fulfill.
Clearing \eqref{parametersCond5} yields \eqref{parametersCond3}.
With this first condition,
doing some, rather tedious, calculations
in eq. \eqref{convergenceConditionSwitching}
the second condition \eqref{parametersCond2} is obtained.
\end{proof}

%In the rest of the section we discuss in detail the
%meaning of the theorem and its implications.
We discuss now in detail the
meaning of the theorem and its implications.
\begin{remark}
Note that the theorem provides just a sufficient condition to ensure
convergence. This means that although the given bounds
seem very restrictive, in practice,
even if we choose large values of $\lambda_M$ and $\lambda_m,$
there will be convergence.
Moreover, an important consequence of corollaries
\ref{CorollaryParam1} and \ref{CorollaryParam2} is that,
independently on the changes of the network topology, there are always
parameters such that the method converges to the consensus.
\end{remark}

%At the view of the results,
\begin{remark}
It is also interesting to note the different behavior
of the algorithm when the topology changes
with respect to the fixed case.
In the latter case, in general it is better to select the parameters
$\lambda_M$ and $\lambda_m$ with large modulus to ensure
that all the eigenvalues of the weight matrix are included
in the interval $[\lambda_m,\lambda_M]$.
However, in the switching case, it is necessary to choose
them small so that $c-d$ is large enough to guarantee convergence.
This happens because the more variation on the eigenvalues
of the weight matrices, the larger $\kappa_1(\max \{ |c\lambda_{\max}-d|, |c\lambda_{\min} - d| \})$ is.
Therefore, the larger $N$, the smaller (in modulus) $\lambda_M$ and $\lambda_m$
should be chosen.
\end{remark}

\begin{remark}
The analysis followed to proof convergence of our algorithm
is also interesting because it can be applied to more general
consensus algorithms based on recurrences of order greater than one.
Given a recurrence similar to \eqref{ChebyshevRecurrenceSwitch},
if a scalar difference equation is found such that its solution bounds
the original one in the worst case,
a convergence result using the behavior of this recurrence can be obtained.
To the authors' knowledge,
this is the first theoretical result proving convergence
of a distributed algorithm based on polynomials
under switching communication topologies.
\end{remark}

%However, note that the theorem provides just a sufficient condition to ensure
%convergence. This implies that although the given bounds
%seem very restrictive, in practice,
%even if we choose large values of $\lambda_M$ and $\lambda_m,$
%there will be convergence.
%Moreover, considering eq. \eqref{parametersCond1} and
%\eqref{parametersCond2} we can see that, independently on
%the changes of the network topology, there are always
%parameters such that the method converges to the consensus.
%To the authors' knowledge,
%it is the first theoretical result proving convergence
%of a distributed algorithm based on the evaluation of polynomials.

Finally, we provide a discussion about the assumptions we have made
to proof convergence.
\begin{itemize}
\item
%Compared to the fixed topology case, here we have required
%the weight matrices to be symmetric.
\emph{Symmetric weight matrices:}
If the weight matrices are not symmetric, then we cannot ensure
that the norm of the matrices used to change the base of eigenvectors is equal to $1$.
In such a case the convergence condition in Theorem \ref{ConvergeSwitchingTopTh}
would be $K\kappa_1(\max \{ |c\lambda_{\max}-d|, |c\lambda_{\min} - d| \}) \tau(c-d) < 1,$ with $K\geq1$ some positive constant.
It is also important to remark that, in this situation, the left eigenvector
associated to $\lambda_1(n)$ is not constant anymore for different matrices.
This makes the theoretical analysis of the behavior more tedious
because at each iteration it is affected by these eigenvectors,
which do not tend to zero with $n$.
However, convergence can still be achieved.

\item
\emph{Connectivity of the graphs:}
The assumption about the connectivity of each graph
is more restrictive than in other approaches, e.g.,~\cite{AJ-JL-ASM:03},
where only joint connectivity is imposed.
In our analysis, if one graph is disconnected, then
$\lambda_{\max}=1$ and the sufficient condition \eqref{convergenceConditionSwitching}
is never satisfied.
This, of course, is caused because we are considering
the worst case scenario, so that we can model the behavior of the
Chebyshev recurrence as the $n$th power of some quantity.
However, in practice, even if some graphs are disconnected,
the errors associated to the
eigenvectors associated to the eigenvalue $1$ are also canceled.
We show this in simulations in section \ref{simulations}.
\end{itemize}

%% file: simulations2.tex
In this section we analyze our algorithm in a simulated environment.
Monte Carlo experiments have been designed to
study the convergence of the method and the influence of
the parameters $\lambda_m$ and $\lambda_M$ in the algorithm.

\subsection{Evaluation with a fixed communication topology}
In a first step we study the algorithm when
the topology of the network is fixed.
We analyze the convergence speed for different weight matrices,
comparing it with other approaches,
and the influence of the parameters $\lambda_M$ and $\lambda_m$
in the performance of the algorithm.

In the experiments we have considered 100
random networks of 100 nodes.
For each network the nodes have been randomly positioned
in a square of $200\times200$ meters.
Two nodes communicate if they are at a distance lower than 20 meters.
The networks are also forced to be connected so that the algorithms converge.
After that, 100 different random initial values
have been generated in the interval $(0,1)^N,$
giving a total of 10000 trials to test the algorithm.

\subsubsection{Convergence speed of the algorithm}
We evaluate how our algorithm behaves compared to other methods
using different weighted adjacency matrices.
For each communication network we have computed 4 different
weighted adjacency matrices.
The first one, $\textbf{A}_{ld},$ uses the ``local degree weights'',
the second one, $\textbf{A}_{bc},$ uses the ``best constant factor''
and the third one, $\textbf{A}_{os},$ computes an approximation
of the ``optimal symmetric weights''.
For more information about these matrices we refer the reader
to~\cite{LX-SB:04}.
These three matrices are symmetric, for that reason we have included
in the experiment a fourth non-symmetric matrix, $\textbf{A}_{ns},$
computed by $a_{ij}=1/(\mathcal{N}_i+1)$ if $j\in\mathcal{N}_i\cup i$
and $a_{ij}=0$ otherwise.

We have compared our method with
the powers of the matrices using \eqref{Consensus_Laplacian},
the Newton's interpolation polynomial of degree 2 proposed in~\cite{EK-PF:09},
$N_2(x) = (x-\alpha)^2/(1-\alpha)^2$,
and the second order recurrence with fixed weights proposed in~\cite{SM-BG-MHS:98},
$F_n(x) = \beta x F_{n-1}(x) + (1-\beta) F_{n-2}(x)$.
%$\textbf{x}(n) = \beta \textbf{A}_{xx} \textbf{x}(n-1) + (1-\beta) \textbf{x}(n-2)$.
We have used the values $\alpha=(\lambda_2+\lambda_N)/2$
and $\beta=2/(1+\sqrt{1-\lambda_2^2})$,
which give the best convergence rate for the two algorithms.
For the Chebyshev polynomials we have also assigned
%We have analyzed the convergence speed of the powers of these
%matrices and the speed of our algorithm using the same matrices,
%$P_n(\textbf{A}_{ld}),\ P_n(\textbf{A}_{bc}),\ P_n(\textbf{A}_{os})$
%and $P_n(\textbf{A}_{ns}).$
%In each case we have assigned
the optimal parameters $\lambda_M=\lambda_2$ and $\lambda_m = \lambda_N$.
%to the algorithm.
We have measured the average number of iterations required to obtain
an error,
$e=\|\textbf{x}(n)-(\textbf{w}_1^T\textbf{x}(0)/\textbf{w}_1^T\textbf{1})\textbf{1}\|_\infty,$ smaller than a given tolerance.

\begin{table}[ht]
\caption{\small Number of iterations for different algorithms and tolerances}
\label{table_comparison_powers} \centering
\begin{tabular}{|c|c|c|c|c||c|c|c|c|c|}
\hline
Method$\backslash$Tolerance&$10^{-2}$&$10^{-3}$&$10^{-4}$&$10^{-5}$&
Method$\backslash$Tolerance&$10^{-2}$&$10^{-3}$&$10^{-4}$&$10^{-5}$\\[2pt]
\hline
$\textbf{A}^n_{ld}$      &396.1   &899.0   &1422.9  &1902.9&
$N_2(\textbf{A}_{ld})$   &381.4   &748.9   &1120.8  &1474.5\\[2pt]
$\textbf{A}^n_{bc}$      &470.5   &892.4   &1307.4  &1691.5&
$N_2(\textbf{A}_{bc})$   &475.7   &897.0   &1109.9  &1493.7\\[2pt]
$\textbf{A}^n_{os}$      &390.8   &735.1   &1092.0  &1446.0&
$N_2(\textbf{A}_{os})$   &426.8   &792.4   &964.3   &1225.2\\[2pt]
$\textbf{A}^n_{ns}$      &308.9   &698.4   &1116.7  &1521.2&
$N_2(\textbf{A}_{ns})$   &302.6   &604.1   &911.5   &1216.4\\[2pt]
$F_n(\textbf{A}_{ld})$   &45.7    &71.9    &98.0    &124.2&
$P_n(\textbf{A}_{ld})$   &41.8    &62.2    &82.6    &103.0 \\[2pt]
$F_n(\textbf{A}_{bc})$   &45.2    &67.4    &91.2    &114.6&
$P_n(\textbf{A}_{bc})$   &44.6    &66.4    &88.1    &109.9\\[2pt]
$F_n(\textbf{A}_{os})$   &42.2    &62.9    &83.3    &103.6&
$P_n(\textbf{A}_{os})$   &42.1    &62.6    &83.0    &103.4\\[2pt]
$F_n(\textbf{A}_{ns})$   &40.8    &63.9    &86.8    &109.8&
$P_n(\textbf{A}_{ns})$   &38.6    &57.1    &75.6    &94.1\\[2pt]
\hline
\end{tabular}
\end{table}

\begin{comment}
Miter =
        49.13       93.341       94.405       37.745        21.46       22.824        21.59       20.083
       396.18       470.58       390.84       308.91       41.865       44.636       42.128       38.638
       899.02       892.49       735.12        698.4       62.276       66.451       62.618       57.171
       1422.9       1307.4         1092       1116.7       82.619       88.168        83.05       75.639
       1902.9       1691.5         1446       1521.2       103.02       109.92       103.48       94.147

Miter =
        64.87       97.753       103.02       49.216       19.491       18.468       17.328       17.842
       381.44       475.71       426.85        302.6       45.743       43.295       39.263       40.898
       748.98       897.03       792.48       604.15       71.918       67.477        60.91       63.908
       1120.8       1309.9       1164.3       911.59       98.094       91.202       82.329       86.872
       1474.5       1693.7       1525.2       1216.4       124.21       114.61       103.61       109.81
\end{comment}
Table \ref{table_comparison_powers} shows the results of the experiment.
%We can observe that
For any matrix our algorithm is the one that reaches the consensus first.
It is remarkable the speed up compared to the powers and the Newton method.
%reaches the desired precision in far less iterations
%than \eqref{Consensus_Laplacian}.
%Using the same matrix, our algorithm obtains
%the same results reducing by one order of magnitude
%the number of iterations.
%(e.g., for $\textbf{A}_{ld}$ and
%a tolerance error of $e<10^{-3},$
%\eqref{Consensus_Laplacian} requires 823.8 iterations and $p_n(\textbf{A}_{ld})$
%only requires 57.5).
Moreover, considering that the initial error is upper bounded by 1,
note that our algorithm is able to reduce the error by five orders of magnitude ($10^{-5}$)
in around $N=100$ iterations ($103.0, 109.9, 103.4$ and $94.1$ iterations in the table),
which is the size of the network.

An interesting detail is that our algorithm
converges faster using the ``local degree weights'', $\textbf{A}_{ld} (103.0),$
and the ``non-symmetric weights'', $\textbf{A}_{ns} (94.1),$
than using the other two matrices $(109.9$ and $103.4)$,
even though the second largest eigenvalue of the other two
matrices is smaller.
This behavior happens because the eigenvalues of
$\textbf{A}_{bc}$ and $\textbf{A}_{os}$ are symmetrically placed
with respect to zero whereas for $\textbf{A}_{ld}$ and $\textbf{A}_{ns}$
$|\lambda_N|<\lambda_2$ (an example can be found in~\cite{LX-SB:04}).
%interval
%that contains all the eigenvalues of $\textbf{A}_{ld}$ and $\textbf{A}_{ns}$
%is smaller than the intervals that contain the eigenvalues of
%$\textbf{A}_{bc}$ and $\textbf{A}_{os}$ (an example can be found in~\cite{LX-SB:04}).
As a consequence, $c-d$ is larger and the algorithm converges faster.
This is indeed very convenient because the ``local degree weights''
and the ``non-symmetric weights''
can be easily computed in a distributed way without global information, whereas the other
two require the knowledge of the whole topology.

Regarding the non-symmetric weights, we have observed that $\lambda_2$ is,
in general, small compared to the second eigenvalue of the symmetric matrices.
Since the eigenvalues of $\textbf{A}_{ns}$ also satisfy that $|\lambda_N|<\lambda_2$,
the convergence for this matrix is the fastest.
Also note that these matrices are the easiest to compute.
On the other hand, when using symmetric weight matrices the convergence value
is known to be the average of the initial conditions whereas when using
non-symmetric weights the convergence value depends on the matrix.

\subsubsection{Dependence on the parameters $\lambda_M$ and $\lambda_m$}
So far we have evaluated the convergence speed of our algorithm
only considering the optimal parameters,
which implies the knowledge of the eigenvalues of the weight matrix.
However, in most situations the nodes will have no knowledge about
these eigenvalues.
We analyze now the convergence rates of our algorithm
when it is run using sub-optimal parameters.
In this case, for simplicity we have only
considered $\textbf{A}_{ld}$ in the experiment.

\begin{table}[ht]
\caption{\small Number of iterations using sub-optimal parameters and tolerance $10^{-3}$}
\label{table_suboptimal_parameters} \centering
\begin{tabular}{|c|c|c|c|c|c|c|}
\hline
$\lambda_m\backslash\lambda_M$&0.2&0.5&0.8&0.9&0.95&0.999\\[2pt]
\hline
-0.2&          713.8&       563.7&       355.2&    $\infty$&    $\infty$&     $\infty$\\[2pt]
-0.5&          798.1&       630.4&       397.2&       279.0&       194.5&         75.9\\[2pt]
-0.8&          874.3&       690.6&       435.2&       305.6&       213.1&         83.1\\[2pt]
-0.9&          898.3&       709.5&       447.0&       314.0&       219.0&         85.4\\[2pt]
-0.95&         910.0&       718.8&       453.0&       318.1&       221.8&         86.5\\[2pt]
-0.999&        919.4&       726.0&       457.6&       321.3&       224.0&         87.4\\[2pt]
\hline
\hline
$F_n$ &        757.5&       672.4&       463.9&       320.9&       227.1&         93.0\\[2pt]
\hline
\end{tabular}
\end{table}
The results are in Table \ref{table_suboptimal_parameters}.
The table shows the average number of iterations required to have
an error lower than $10^{-3}.$
The number of iterations is in all the cases larger than in
Table \ref{table_comparison_powers} ($62.2$ iterations)
but anyway,
the results are in most cases also good.
The only problem appears when $\lambda_M+\lambda_m-1>\lambda_N$
because the algorithm diverges (cells with $\infty$ in the table).
Nevertheless, the number of iterations is almost always smaller than using the powers
of $\textbf{A}_{ld}$ and the Newton polynomial ($899.0$ and $748.9$ iterations
in Table \ref{table_comparison_powers} respectively).
The results compared to $F_n$ evaluated with
the optimal parameter ($71.9$ it. in Table \ref{table_comparison_powers}) seem to be poor.
However, the optimal $\beta$ requires the knowledge of $\lambda_2$
which, right now, we are assuming it is unknown.
For that reason, in the last row of Table \ref{table_suboptimal_parameters}
we have included the results using $F_n$ evaluated with
$\beta=2/(1+\sqrt{1-\lambda_M^2})$,
i.e., with the same estimation
of $\lambda_2$ used for the Chebyshev polynomials.
In this case we observe again that both methods
present a similar performance when using the same parameters.
The degree of freedom given by $\lambda_m$ is what differs in the algorithms.
By adjusting this parameter we can reduce the number
of iterations in our algorithm.

Another advantage of using our algorithm with the
weight matrix $\textbf{A}_{ld},$
besides the computation using local information,
is that usually its smallest eigenvalue, $\lambda_N,$
is a negative value close to zero
(in our simulations it has never valued less than -0.5).
The second largest eigenvalue depends on how many nodes has the network
and the number of links, but in general
this eigenvalue is close to one.
Therefore by choosing $\lambda_m = -0.5$ and $\lambda_M \simeq 1$
there is a great chance to obtain a good convergence rate
and almost no risk of divergence,
see for example the cell in the second row and sixth column of
Table \ref{table_suboptimal_parameters} ($153.7$).
A safer choice of parameters is $\lambda_m = -\lambda_M,$
which we know that has good convergence rates.
In this case it is also convenient to choose $\lambda_M \simeq 1$
to ensure that all the eigenvalues are contained in $[\lambda_m,\lambda_M].$

\subsection{Evaluation with a switching communication topology}
Let us see how the algorithm behaves when the
topology of the network changes at different iterations.
We start by showing the convergence in an illustrative example
where the conditions of Theorem \ref{ConvergeSwitchingTopTh} are satisfied.
After that we run again Monte Carlo experiments to analyze
the algorithm in more realistic situations, where the conditions of
Theorem \ref{ConvergeSwitchingTopTh} do not always hold.

\subsubsection{Illustrative Example}
The communication network considered, composed by 20 nodes, is depicted in Fig.
\ref{IllustrativeExampleFig} (top left), which is connected.
In order to satisfy the conditions of Theorem \ref{ConvergeSwitchingTopTh}
at each iteration we have randomly added some links to the network.
In this way all the topologies remain connected and the parameters
$\lambda_{\max}$ and $\lambda_{\min}$ correspond to the second maximum
and the smallest eigenvalues of the initial weight matrix.
Using the local degree weights, which return a symmetric matrix,
these parameters are $\lambda_{\max}=0.9477$ and $\lambda_{\min}=-0.1922$.
%The evolution of $x(n)$ using \eqref{Consensus_Laplacian} is shown in the top left figure.
Figure \ref{IllustrativeExampleFig} top middle and top right
depict the evolution of $x(n)$ using \eqref{ChebyshevRecurrenceSwitch}
with the parameters of Corollary \ref{CorollaryParam1}, $\lambda_M = -\lambda_m = 0.3190$,
and Corollary \ref{CorollaryParam2},
$\lambda_M=0.6274,\ \lambda_m=0.1282$, respectively.
The evolution of $\textbf{x}(n)$ using \eqref{Consensus_Laplacian} is shown
in Fig. \ref{IllustrativeExampleFig} bottom left.
It is interesting to note the similarity of this graphic with the
Chebyshev recurrence using the symmetric parameters given by Corollary \ref{CorollaryParam1} (top middle).
Finally, to remark that the condition of Theorem \ref{ConvergeSwitchingTopTh}
is a sufficient condition in Fig. \ref{IllustrativeExampleFig}
bottom middle and bottom right we show that the algorithm also
converges to the consensus choosing parameters with larger modulus.
In the example we have chosen the parameters using the criteria analyzed for the
fixed topology situation.
Moreover, we can see in the graphics that the consensus is achieved in both cases in less iterations
(the lines overlap earlier in the graphics).
Finally, note that the symmetry in all the weight matrices implies that,
in all the cases, the value of the consensus is the average of the initial conditions.
\begin{figure*}
  \centering
  \begin{tabular}{ccc}
 \includegraphics[width=0.20\linewidth]{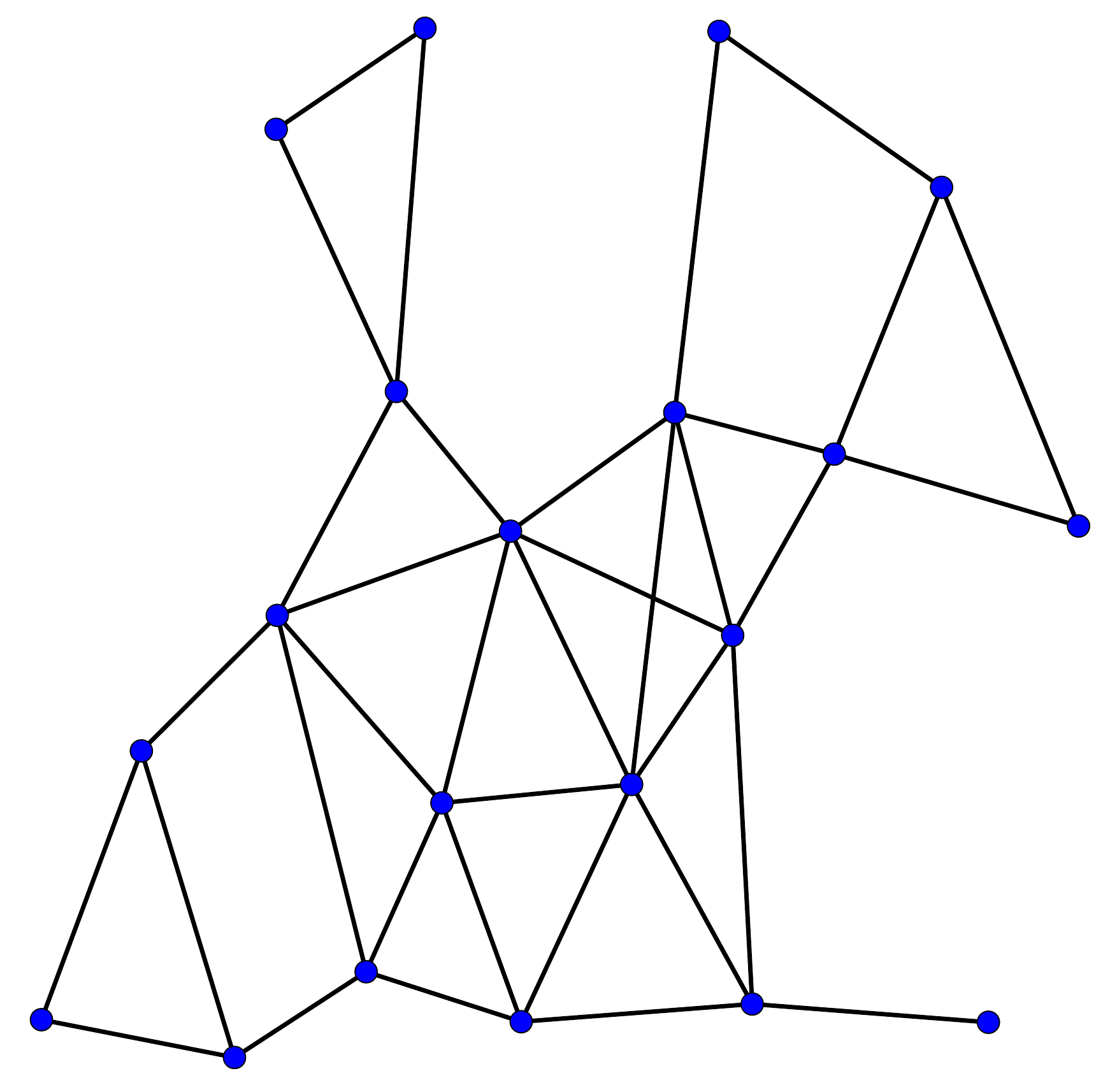}&
 \includegraphics[width=0.32\linewidth]{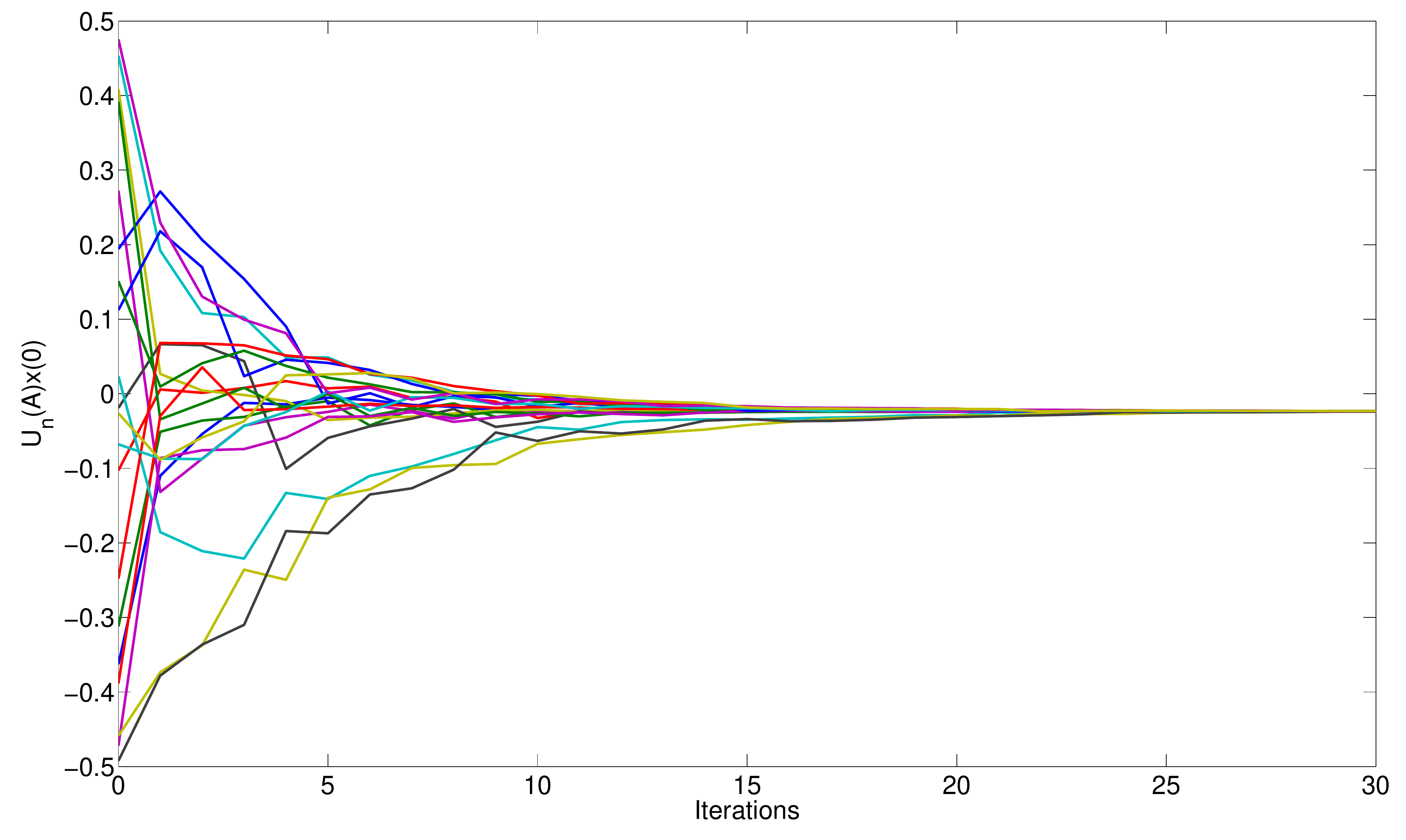}&
 \includegraphics[width=0.32\linewidth]{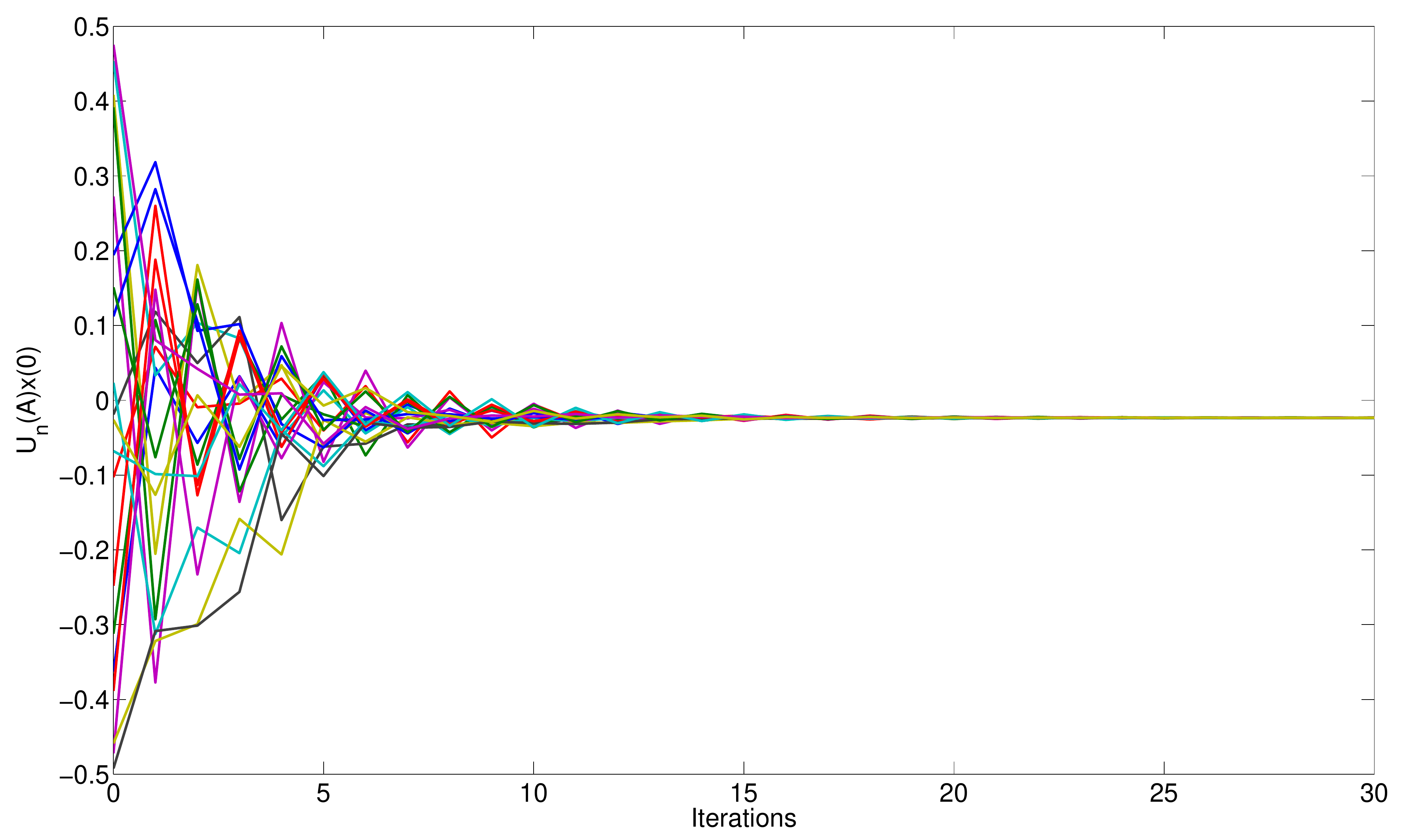}\\
  Communication Network&
  $\lambda_M = -\lambda_m = 0.3190$&
  $\lambda_M=0.6274,\ \lambda_m=0.1282$\\
  \includegraphics[width=0.32\linewidth]{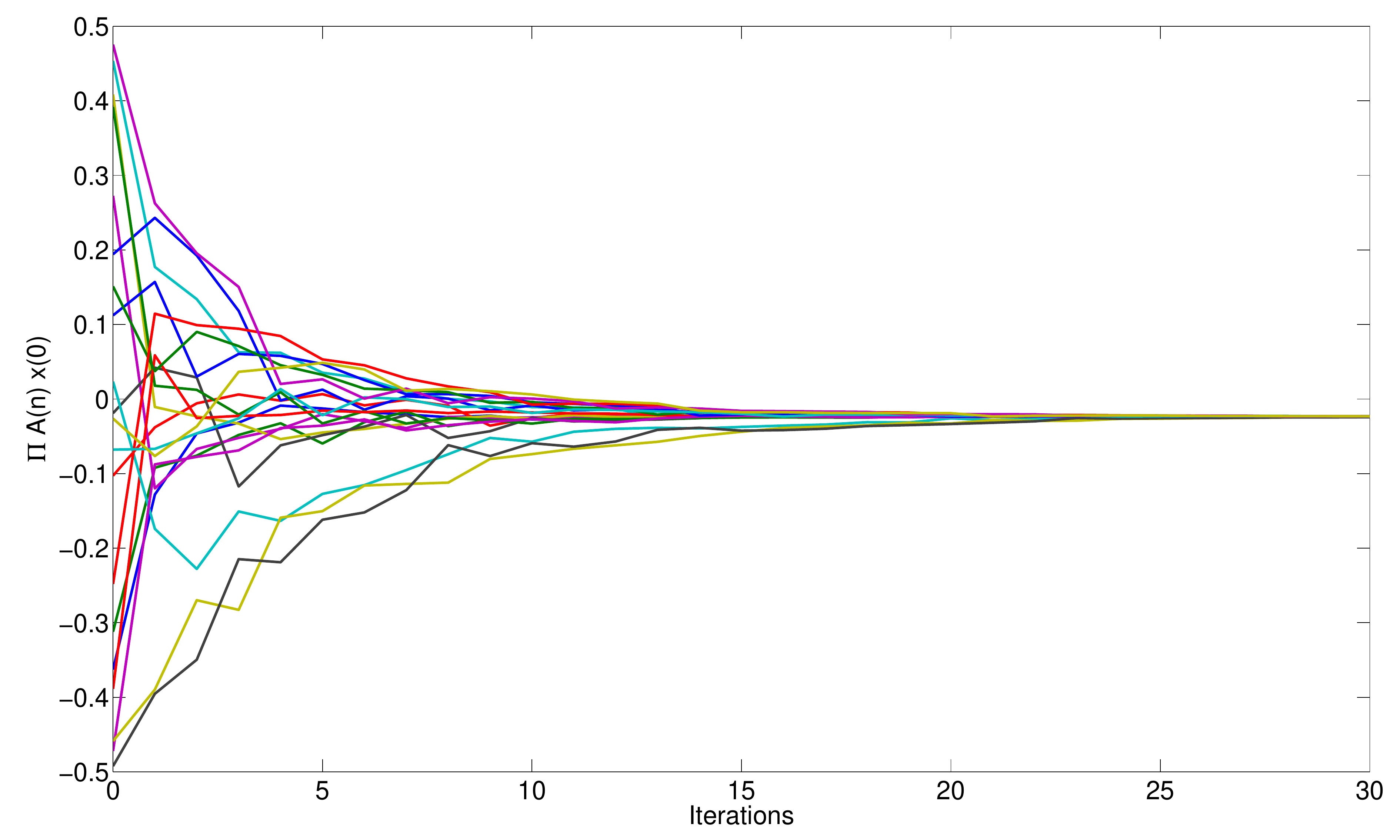}&
 \includegraphics[width=0.32\linewidth]{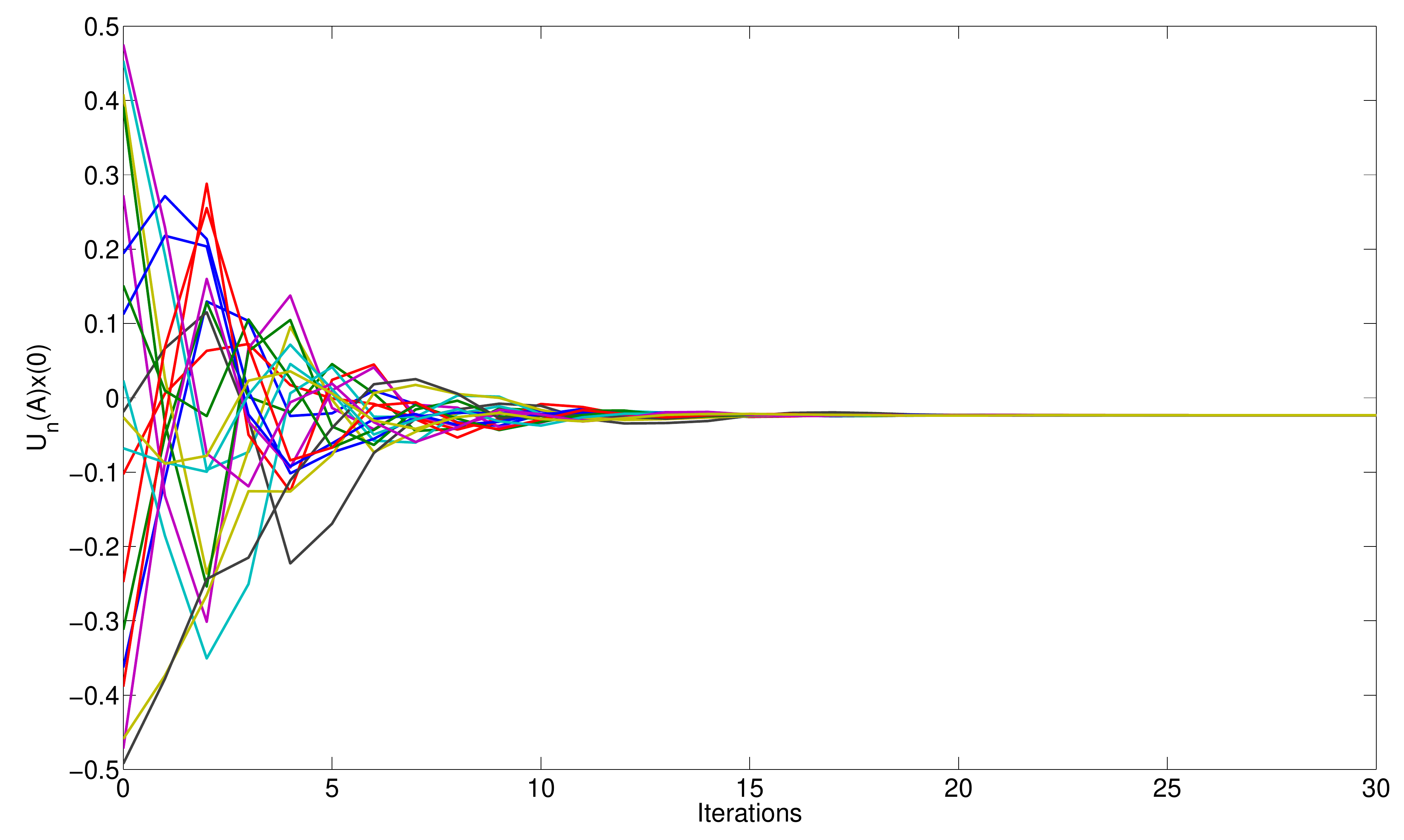}&
 \includegraphics[width=0.32\linewidth]{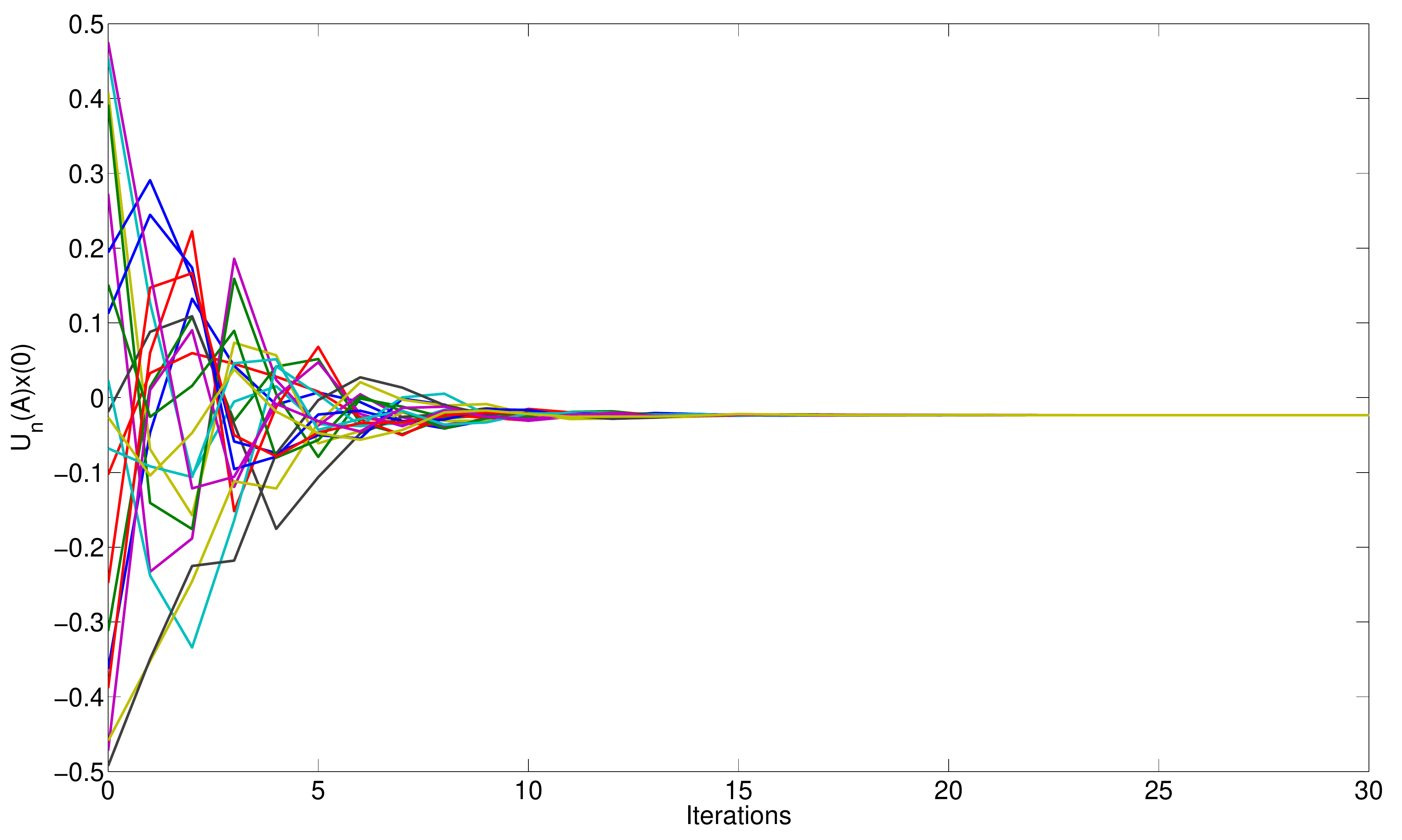}\\
  $\Pi_n \textbf{A}(n)\textbf{x}(0)$&
  $\lambda_M=0.8,\ \lambda_m=-0.8$&
  $\lambda_M=0.9,\ \lambda_m=-0.5$
  \end{tabular}
  \caption{Illustrative example of the convergence speed of the algorithm
  with a switching communication topology.
  The initial network is shown at the top left graphic.
  The evolution using \eqref{Consensus_Laplacian} is shown at the bottom left
  and four different executions of \eqref{ChebyshevRecurrenceSwitch}
  with the same changes in the topology and different parameters
  are depicted in the rest of the graphics. Notice that even when the conditions
  of Theorem \ref{ConvergeSwitchingTopTh} are not satisfied (bottom middle and bottom right graphics),
  the algorithm still achieves the consensus.
  }
  \label{IllustrativeExampleFig}
\end{figure*}

\subsubsection{Analysis of convergence depending on the evolution of network and the parameters
of the algorithm}
\begin{comment}
Random Network
xx =

   12.7381   13.1085   19.6706   71.1459       Inf
   13.2064   12.8969   17.7095   43.7363  266.4952
   13.7809   13.0239   17.0446   36.6125  105.2806
   14.1414   13.1759   16.9186   34.6240   85.3825
   14.2618   13.2338   16.8918   34.1398   81.2287
Miter =
   13.6959

EVOLUTION
xx =

  499.7320  298.5506  198.1569  125.3963   81.8005
  479.9560  319.0362  214.3118  138.3504   90.1594
  482.5467  342.4559  229.3373  149.6142   98.5891
  489.0352  339.3222  237.7992  155.5602  102.6776
  492.2408  342.0014  240.5471  157.4218  104.0553
Miter =
  336.5861

Link Failures
xx =

       1519.7       479.52       184.02       114.03       77.775
       1342.3       474.35       201.69       124.96       85.172
         1228       465.45       217.95       135.01       91.998
       1200.4       470.46       227.14       140.71        95.89
       1195.2       472.29       230.12       142.56       97.161
Miter =
       332.46
   \end{comment}

We have generated again 100 random networks of 100 nodes like in the
fixed topology case.
To model the changes in the communication topology
we have considered three different scenarios in the experiment.
The first one assumes a fixed initial communication topology
and, at each iteration the links can fail with constant probability
equal to $0.05$ (Link Failures).
This is a usual way to model networks with unreliable or noisy communications.
In the second scenario we consider a set of mobile agents that randomly move
in the environment. In this way, at each iteration the
communication topology evolves with the proximity graph
defined by the new positions of the agents (Evolution with Motion).
The last scenario assumes a new random network
at each iteration (Random Network).
Although in reality this situation will be uncommon,
it is interesting to analyze it in order to study the properties
of our algorithm.
In the three scenarios we have used the local degree weights
to define the weight matrix at each iteration.
We have not worried about the network connectivity,
letting the experiment to possibly have several iterations
with disconnected networks.
We have set a maximum of 3000 iterations per trial.

Table \ref{table_switching_laplacian} shows the number of
iterations required by iteration \eqref{Consensus_Laplacian}
to achieve a precision of $10^{-3}$.
We can see that when the network has link failures or evolves
with the motion of the nodes the number of iterations required by
the algorithm is slightly greater than when the topology of the network
remains fixed (1087.2 and 1032.4 compared to 899.0 in Table \ref{table_comparison_powers}).
On the other hand, when the network randomly changes at each step,
in a few iterations (9.4) the consensus is achieved, which makes
sense because in this situation the information is spread in a fast way.
\begin{table}[!ht]
\caption{\small Number of iterations with tolerance $10^{-3}$}
\label{table_switching_laplacian} \centering
\begin{tabular}{|c|c|c|}
\hline
Link Failures   & Evolution with Motion  & Random Networks\\[2pt]
\hline
1087.2  &1032.4  & 9.4\\[2pt]
\hline
\end{tabular}
\end{table}

The number of iterations required to achieve the same accuracy (tolerance of $10^{-3}$)
using \eqref{ChebyshevRecurrenceSwitch} with different parameters
is shown in Tables
\ref{table_switching_linkfailures},
\ref{table_switching_motionevolution} and
\ref{table_switching_randomnetwork}
for the Link Failures, Evolution with Motion and Random Networks
scenarios respectively.
\begin{table}[!ht]
\caption{\small Number of iterations for Link Failures}
\label{table_switching_linkfailures} \centering
\begin{tabular}{|c|c|c|c|c|c|}
\hline
$\lambda_m\backslash\lambda_M$&0.25&0.5&0.75&0.9&0.95\\[2pt]
\hline
-0.25   &$\geq3000$       &$\geq3000$       &1298.1       &383.3       &267.9\\[2pt]
-0.5    &$\geq3000$       &$\geq3000$       &1328.6       &418.9       &293.5\\[2pt]
-0.75   &$\geq3000$       &$\geq3000$       &1356.6       &452.3       &316.8\\[2pt]
-0.9    &$\geq3000$       &$\geq3000$       &1321.0       &470.9       &330.0\\[2pt]
-0.95   &$\geq3000$       &$\geq3000$       &1326.4       &476.9       &334.5\\[2pt]
\hline
\end{tabular}
\end{table}
\begin{table}[!ht]
\caption{\small Number of iterations for Evolution with Motion}
\label{table_switching_motionevolution} \centering
\begin{tabular}{|c|c|c|c|c|c|}
\hline
$\lambda_m\backslash\lambda_M$&0.25&0.5&0.75&0.9&0.95\\[2pt]
\hline
-0.25   &$\geq3000$  &1738.0  &600.1  &457.2  &260.9\\[2pt]
-0.5    &$\geq3000$  &1765.2  &665.6  &461.9  &306.5\\[2pt]
-0.75   &1726.5      &1793.5  &703.6  &506.3  &309.8\\[2pt]
-0.9    &1740.0      &1813.0  &708.5  &564.9  &311.0\\[2pt]
-0.95   &1744.5      &1818.0  &710.4  &564.9  &311.5\\[2pt]
\hline
\end{tabular}
\end{table}
\begin{table}[!ht]
\caption{\small Number of iterations for Random Networks}
\label{table_switching_randomnetwork} \centering
\begin{tabular}{|c|c|c|c|c|c|}
\hline
$\lambda_m\backslash\lambda_M$&0.25&0.5&0.75&0.9&0.95\\[2pt]
\hline
-0.25   &8.1   &8.3    &11.8   &25.4  &$\infty$\\[2pt]
-0.5    &8.3   &8.9    &11.6   &22.3  & 42.1 \\[2pt]
-0.75   &8.7   &9.6    &11.8   &21.7  & 37.8 \\[2pt]
-0.9    &8.9   &10.0   &12.0   &21.7  & 36.8 \\[2pt]
-0.95   &9.0   &10.1   &12.0   &21.7  & 36.5 \\[2pt]
\hline
\end{tabular}
\end{table}

With these results we can extract some interesting remarks.
First of all, for the parameters tested in the experiment,
the algorithm is convergent in almost all the cases.
Only in the Random Networks the algorithm diverges when
$\lambda_M=0.95$ and $\lambda_m=-0.25$
(Table \ref{table_switching_randomnetwork} first row and sixth column).
The cells with ``$\geq3000$'' iterations point that for these
parameters the algorithm converges but in a slow way.
A second interesting detail is that, similarly to the fixed topology
case, we can always find parameters that make our algorithm
achieve the consensus faster than using \eqref{Consensus_Laplacian}
(results of Table \ref{table_switching_laplacian}).
However, it is surprising which parameters achieve this goal
in the different scenarios.
For the Link Failures and the Evolution with Motion, the best parameters
are exactly the parameters that make the algorithm
diverge for the Random Networks scenario,
i.e., $\lambda_M=0.95$ and $\lambda_m=-0.25$ with 267.9
and 260.9 iterations respectively.
On the other hand, the best parameters for the Random Networks are those
who give the slowest convergence rate for the other two scenarios,
i.e., $\lambda_M=0.25$ and $\lambda_m=-0.25$ with 8.1 iterations in Table \ref{table_switching_randomnetwork}
versus more than 3000 in Tables \ref{table_switching_linkfailures} and \ref{table_switching_motionevolution}.
The explanation for this phenomenon appears in the variability of the eigenvectors
of the weight matrices.
When the topology changes arbitrarily at each iteration, there is
a great variability in the eigenvectors of the weight matrices,
which turns out in a great variability of $\textbf{x}(n)$.
This situation is closer to the worst case we have shown in section IV
to proof the convergence of the algorithm.
Therefore, a good convergence rate requires a large value of $c-d,$
achieved when $\lambda_M$ and $\lambda_m$ have small modulus.
When the topology changes smoothly, as in the Link Failures and the Motion Evolution,
the eigenvectors almost do not change and the algorithm behaves similarly to the fixed case.
For that reason, the parameters that achieve the best convergence rate are the same
as in the fixed case.
However, we must be careful because for larger values of $\lambda_M$ the algorithm
may diverge.

A final detail is that, in all the cases, the convergence seems
to be more affected by $\lambda_M$ than $\lambda_m.$
This is explained by the use of the local degree weights.
As we have mentioned earlier, these matrices do not have symmetric eigenvalues
with respect to zero.
In these matrices $\lambda_{\max}$ dominates the convergence rate,
so the convergence is more sensible to
the parameter $\lambda_M$.
%Again, we can see this as an additional benefit of using these matrices
%for the consensus (besides they can be computed in a distributed way)
%combined with our approach.

In conclusion, when the topology of the network changes,
the parameters should be chosen taking into account the nature of these changes.
For small changes similar parameters to the fixed case should be assigned
whereas if the network is expected to change a lot we should pick
small parameters for the algorithm to guarantee convergence.

%% file: proof_convergenceFixed.tex
We introduce two auxiliary results to proof the convergence.
%%%%%%%%%%%%%%%%%%%%%%%%%%%%%%%%%%%%%%%%%%%%%%%%%%%%%%%%%%%%%%
%LEMMA
%%%%%%%%%%%%%%%%%%%%%%%%%%%%%%%%%%%%%%%%%%%%%%%%%%%%%%%%%%%%%%
\begin{lemma}
\label{lemma2}
Given $x_1 > 1,$
for any $x_2$ such that $|x_2| < x_1$ it holds that
%$\lim_{n\to \infty} \dfrac{T_n(x_2)}{T_n(x_1)} =0.$
\begin{equation}
\label{limitTozero}
\lim_{n\to \infty} \dfrac{T_n(x_2)}{T_n(x_1)} =0.
\end{equation}
\end{lemma}
\vskip 0.2cm
%%%%%%%%%%%%%%%%%%%%%%%%%%%%%%%%%%%%%%%%%%%%%%%%%%%%%%%%%%%%%%
%PROOF
%%%%%%%%%%%%%%%%%%%%%%%%%%%%%%%%%%%%%%%%%%%%%%%%%%%%%%%%%%%%%%
\begin{proof}
For $|x_2|\leq1,$ $|T_n(x_2)|\leq1,\ \forall n,$ and since $T_n(x_1)\to \infty$
with $n,$
eq. \eqref{limitTozero} is true.
Now, if $1<|x_2| < x_1,$ then
using \eqref{direct_Expression} we have
\begin{equation}
\label{fractionTn}
\dfrac{T_n(x_2)}{T_n(x_1)} = \dfrac{\tau(x_1)^{n}}{\tau(x_2)^{n}} \;
\dfrac{1+\tau(x_2)^{2n}}{1+\tau(x_1)^{2n}}.
%\dfrac{\tau(b)^{n}}{\tau(a)^{n}} \;\dfrac{\dfrac{1}{\tau(b)^{2n}}+
%1}{\dfrac{1}{\tau(a)^{2n}}+ 1},
\end{equation}
%and if $1< |x_2| < x_1$, then
But in this case $1>|\tau(x_2)|>\tau(x_1)>0$ and the result holds immediately.
\end{proof}
%%%%%%%%%%%%%%%%%%%%%%%%%%%%%%%%%%%%%%%%%%%%%%%%%%%%%%%%%%%%%%
%END PROOF
%%%%%%%%%%%%%%%%%%%%%%%%%%%%%%%%%%%%%%%%%%%%%%%%%%%%%%%%%%%%%%
%%%%%%%%%%%%%%%%%%%%%%%%%%%%%%%%%%%%%%%%%%%%%%%%%%%%%%%%%%%%%%
%LEMMA
%%%%%%%%%%%%%%%%%%%%%%%%%%%%%%%%%%%%%%%%%%%%%%%%%%%%%%%%%%%%%%
\begin{lemma}
\label{lemma3}
Given $x >1$, for any complex number $z,$ such that
$|\tau(z)| = \min\{|z+\sqrt{z^2-1}|, |z-\sqrt{z^2-1}|\} > \tau(x),$ then
$\lim_{n\to \infty} T_n(z)/T_n(x) =0.$
\end{lemma}
\vskip 0.2cm
%%%%%%%%%%%%%%%%%%%%%%%%%%%%%%%%%%%%%%%%%%%%%%%%%%%%%%%%%%%%%%
%PROOF
%%%%%%%%%%%%%%%%%%%%%%%%%%%%%%%%%%%%%%%%%%%%%%%%%%%%%%%%%%%%%%
\begin{proof}
It is a straightforward consequence of \eqref{fractionTn}.
\end{proof}
\textbf{Proof of Theorem \ref{ConvergenceTh}.}
%Let $\textbf{Q}= \textbf{A}- \frac{1}{N}\textbf{1}\textbf{1}^T$,
Let $\textbf{Q}= \textbf{A}- \textbf{1}\textbf{w}_1^T/\textbf{w}_1^T\textbf{1}$,
whose eigenvalues are $0,$ with $\textbf{v}_1$ its corresponding right eigenvector,
and $\lambda_2, \ldots, \lambda_N$ with the same eigenvectors as $\textbf{A}$.
Since $\textbf{v}_1 = \textbf{w}_1^T\textbf{x}(0)\textbf{1}/\textbf{w}_1^T\textbf{1}$,
then $\textbf{1}\textbf{w}_1^T(\textbf{x}(0)-\textbf{v}_1) = 0.$
Taking this into account it is easy to see that
\begin{equation}
\textbf{A}^n (\textbf{x}(0) - \textbf{v}_1) = \textbf{Q}^n (\textbf{x}(0) - \textbf{v}_1),\quad \forall n \in \mathbb{N},
\end{equation}
and therefore $P_n(\textbf{A})(\textbf{x}(0)-\textbf{v}_1) = P_n(\textbf{Q})(\textbf{x}(0)-\textbf{v}_1).$

Also $\textbf{A} \textbf{v}_1= \textbf{v}_1$ and
$P_n(1)=1$, then $P_n(\textbf{A})\textbf{v}_1=\textbf{v}_1$ and
%\begin{equation}
%\begin{split}
%&\Vert \textbf{x}(n) - \textbf{v}_1 \Vert_2 =
%\Vert P_n(\textbf{A})(\textbf{x}(0) -\textbf{v}_1 )\Vert_2=\\
%&\Vert P_n(\textbf{Q}) (\textbf{x}(0) -\textbf{v}_1) \Vert_2
%\le \Vert P_n(\textbf{Q}) \Vert_2 \Vert \textbf{x}(0) -\textbf{v}_1 \Vert_2.
%\end{split}
%\end{equation}
\begin{equation}
\Vert \textbf{x}(n) - \textbf{v}_1 \Vert_2 =
\Vert P_n(\textbf{A})(\textbf{x}(0) -\textbf{v}_1 )\Vert_2=\\
\Vert P_n(\textbf{Q}) (\textbf{x}(0) -\textbf{v}_1) \Vert_2
\le \Vert P_n(\textbf{Q}) \Vert_2 \Vert \textbf{x}(0) -\textbf{v}_1 \Vert_2.
\end{equation}
In addition, since $\textbf{A}$ is diagonalizable,
so is $\textbf{Q},$
which implies that $\textbf{Q}$
can be decomposed, $\textbf{Q}=\textbf{P}\textbf{D}\textbf{P}^{-1},$
with $\textbf{D} = $diag$(0,\lambda_2,\ldots,\lambda_N)$.
Using algebra rules we get that
$P_n(\textbf{Q})= \textbf{P} P_n(\textbf{D}) \textbf{P}^{-1}$
and then
\begin{equation}
%\small
\Vert P_n(\textbf{Q}) \Vert_2 \leq
\Vert \textbf{P} \Vert_2\ \rho(P_n(\textbf{Q}))\ \Vert \textbf{P}^{-1} \Vert_2 =
K \max_{i\ne 1} |P_n(\lambda_i)| =
K \max_{i\ne 1} \dfrac{|T_n(c \lambda_i -d)|}{T_n(c-d)},
\end{equation}
with $K$ the condition number of $\textbf{P}.$
%which implies that its spectral norm
%coincides with the spectral radius,
%\begin{equation}
%\small
%\Vert P_n(\textbf{Q}) \Vert_2= \rho(p_n(\textbf{Q})) = \max_{i\ne 1} |p_n(\lambda_i)| = \max_{i\ne 1} \dfrac{|T_n(c %\lambda_i -d)|}{T_n(c-d)}.
%\end{equation}

For any $x \in(\lambda_M+\lambda_m-1,1)$
we have that $|cx -d| < c-d,$ then
for all the real eigenvalues of $\textbf{A}$ but $\lambda_1,$
$|c\lambda_i -d| < c-d.$
Noting that $c-d$ is strictly larger than 1
and $\tau(c-d)<\tau(c\lambda_z-d),$
for any complex eigenvalue $\lambda_z$,
by Lemmas \ref{lemma2} and \ref{lemma3},  $p_n(\lambda_i) \to 0$
for all $i\ne 1$, which proves the convergence of the algorithm.

%Now, if $\lambda_i \in [\lambda_m, \lambda_M]$ for all $i\ne 1$, $|c\lambda_i -d| \le 1$ and
%\begin{equation}
%%\max_{i\ne 1} |p_n(\lambda_i)| \le \dfrac{1}{T_n(c-d)}.
%\max_{i\ne 1} |p_n(\lambda_i)| \le 1/T_n(c-d).
%\end{equation}
%The value in \eqref{convRate} is extracted from this and Lemma \ref{lemma1}.
\EndProof 

%% file: proof_optimalParameters.tex
In order to proof Theorem \ref{Th_optimalParameters}
we will use the following auxiliary results.
\begin{lemma}
\label{lemmaOptimal1}
Let $\lambda_m,\ \lambda_M$ such that
$[\lambda_N, \lambda_2] \not\subseteq [\lambda_m, \lambda_M]$
and $|c\lambda_N-d| <c\lambda_2 -d$.
Then, for fixed $c$, $\nu(c,d)$ is a decreasing function of $d$.
\end{lemma}
\begin{proof}
Let us see that $\partial \nu(c,d) /\partial d <0$.
\begin{equation*}
\nu(c,d)= \dfrac{\tau(c -d)}{|\tau(c\lambda_2 -d)|}=\dfrac{\tau(c -d)}{\tau(c\lambda_2 -d)}>0
\end{equation*}
Then
\begin{equation*}
\dfrac{\partial\nu}{\partial d}= \dfrac{-\tau'(c -d) \tau(c\lambda_2 -d) + \tau(c -d) \tau'(c\lambda_2 -d)}{\tau(c\lambda_2 -d)^2}.
\end{equation*}
But since for $x>0$, $\tau'(x) = -\tau(x) / \sqrt{x^2-1},$ then
\begin{equation*}
\dfrac{\partial\nu}{\partial d}= \dfrac{\tau(c -d)}{\tau(c\lambda_2 -d)}
 \left[\dfrac{1}{\sqrt{(c-d)^2-1}} - \dfrac{1}{\sqrt{(c\lambda_2 -d)^2-1}}
\right]
\end{equation*}
which is negative because $1 < (c\lambda_2 -d)^2 < (c-d)^2$.
\end{proof}
\vskip 0.5pt
%%%%%%%%%%%%%%%%%%%%%%%%%%%%%%%%%%%%%%%%%%%%%%%%%%%%%%%%%%%%%%%%%%%%
%% END OF PROOF LEMMA
%%%%%%%%%%%%%%%%%%%%%%%%%%%%%%%%%%%%%%%%%%%%%%%%%%%%%%%%%%%%%%%%%%%%

\begin{lemma}
\label{lemmaOptimal2}
Let $\lambda_m,\ \lambda_M$ such that
$[\lambda_N, \lambda_2] \not\subseteq [\lambda_m, \lambda_M]$
and $|c\lambda_N-d| >|c\lambda_2 -d|$ with $c\lambda_N-d <0$.
Then, for fixed $c$, $\nu(c,d)$ is an increasing function of $d$.
\end{lemma}
\begin{proof}
Let us see that $\partial \nu(c,d) /\partial d >0$.
\[
\nu(c,d)= \dfrac{\tau(c -d)}{|\tau(c\lambda_N -d)|}=\dfrac{\tau(c -d)}{-\tau(c\lambda_N -d)}>0
\]
Then
\[
\dfrac{\partial\nu}{\partial d}= \dfrac{\tau'(c -d) \tau(c\lambda_N -d) - \tau(c -d) \tau'(c\lambda_N -d)}{\tau(c\lambda_N -d)^2}
\]
But since, for $x<0$,  $\tau'(x) = \tau(x) / \sqrt{x^2-1},$ then
\begin{equation*}
\small
\frac{\partial\nu}{\partial d}= \frac{\tau(c -d)}{ -\tau(c\lambda_N -d)}
\left[\frac{1}{\sqrt{(c-d)^2-1}} + \frac{1}{\sqrt{(c\lambda_2 -d)^2-1}}
\right]
\end{equation*}
which is positive.
\end{proof}
\vskip 0.5pt

\begin{proposition}
Let $\lambda_m,\ \lambda_M$ such that $\lambda_M-\lambda_m=2/c$ is fixed and
$[\lambda_N, \lambda_2] \not\subseteq [\lambda_m, \lambda_M]$. Then
\begin{itemize}
\item[i)]
If $\lambda_2- \lambda_N > \lambda_M- \lambda_m$,
$\nu(c,d) \ge \nu(c,d^*)$, $d^*$ being the value such that $\lambda_M+\lambda_m=\lambda_2+\lambda_N$,
that is, for a fixed $c$, $\nu(c,d)$ is minimum when $\lambda_m, \lambda_M$ are symmetrically  placed with respect to $\lambda_N, \lambda_2$.
\item[ii)]
If $\lambda_2- \lambda_N \le \lambda_M- \lambda_m$ and $\lambda_M < \lambda_2$
then $\nu(c,d) \ge \nu(c,d^*)$, $d^*$ being such that $\lambda_M=\lambda_2$,
and in this case $[\lambda_N, \lambda_2] \subseteq [\lambda_m, \lambda_M]$
\item[iii)]
If $\lambda_2- \lambda_N \le \lambda_M- \lambda_m$ and $\lambda_m > \lambda_N$
then $\nu(c,d) \ge \nu(c,d^*)$, $d^*$ being such that $\lambda_m=\lambda_N$,
and in this case $[\lambda_N, \lambda_2] \subseteq [\lambda_m, \lambda_M]$
\end{itemize}
\end{proposition}
\begin{proof}
\begin{itemize}
\item[i)]
The result follows from Lemmas \ref{lemmaOptimal1} and \ref{lemmaOptimal2}.
If $\lambda_2 >\lambda_M$, then $c\lambda_2-d > |c\lambda_N-d|$
and $\nu(c,d)$ is a decreasing function of
$d=(\lambda_M+\lambda_m)c/2$ which means that
it decreases as $\lambda_M$ increases.
The maximum value of $\lambda_M$
for which these conditions hold is
$\lambda_M=1/c + (\lambda_2+\lambda_N)/2$
for which  $c\lambda_2-d = |c\lambda_N-d|$.

If $\lambda_N  < \lambda_m$, then $c\lambda_2-d < |c\lambda_N-d|$ and $\nu(c,d)$ is an increasing function of $d=(\lambda_M+\lambda_m)c/2$ which means that it increases when $\lambda_M$ increaseses.  The minimum value of $\lambda_M$ for which these conditions hold is
$\lambda_M=1/c + (\lambda_2+\lambda_N)/2$ for which  $c\lambda_2-d = |c\lambda_N-d|$.
\item[ii)]
In this case $c\lambda_2-d > |c\lambda_N-d|$, and $\nu(c,d)$ is a decreasing function of $d=(\lambda_M+\lambda_m)c/2$ which means that it decreases when  $\lambda_M$ increases.
The maximum value of $\lambda_M$ for which these conditions hold is
$\lambda_M=\lambda_2$.
\item[iii)]
In this case $c\lambda_2-d < |c\lambda_N-d|$, and $\nu(c,d)$ is an increasing function of $d=(\lambda_M+\lambda_m)c/2$ which means that it increases when $\lambda_m$ increases.
The minimum value of $\lambda_m$ for which these conditions hold is
$\lambda_m=\lambda_N$.
\end{itemize}
\end{proof}

And finally, we are able to proof the theorem.

\textbf{Proof of Theorem \ref{Th_optimalParameters}.}
If $[\lambda_2, \lambda_N] \subseteq [\lambda_m, \lambda_M]$
the result was proved in~\cite{EM-JIM-CS:11}.
Let us suppose then that $[\lambda_2, \lambda_N] \not\subseteq [\lambda_m, \lambda_M]$.
If $\lambda_2- \lambda_N \le \lambda_M- \lambda_m$, it has been shown in Proposition 1.1
that $\nu(c,d)$ has smaller values for $c,d$ such that
$[\lambda_N, \lambda_2] \subseteq [\lambda_m, \lambda_M]$, and in this case
$\lambda_2=\lambda_M$ and $\lambda_N=\lambda_m$ yields to the minimum
$\nu(c,d)$.

If $\lambda_2- \lambda_N > \lambda_M- \lambda_m$, we have seen in Proposition 1.1
that $\nu(c,d)$ is smaller for $c,d$ such that
$\lambda_m, \lambda_M$ are symmetrically  placed with respect to $\lambda_N, \lambda_2$,
that is, $\lambda_M= \lambda_2 - \alpha$ and $\lambda_m=\lambda_N+\alpha$, $\alpha \ge 0$.  Let us see
that $\nu(c,d)$ is minimum for $\alpha=0$.
First, note that
%\begin{equation*}
%\begin{split}
%&c=\dfrac{2}{\lambda_M-\lambda_m}=\dfrac{2}{\lambda_2-\lambda_N-2\alpha},\\
%&d=\dfrac{\lambda_M+\lambda_m}{\lambda_M-\lambda_m}=
%\dfrac{\lambda_2+\lambda_N}{\lambda_2-\lambda_N-2\alpha}.
%\end{split}
%\end{equation*}
\begin{equation*}
c=\dfrac{2}{\lambda_M-\lambda_m}=\dfrac{2}{\lambda_2-\lambda_N-2\alpha},\hbox{ and }
d=\dfrac{\lambda_M+\lambda_m}{\lambda_M-\lambda_m}=
\dfrac{\lambda_2+\lambda_N}{\lambda_2-\lambda_N-2\alpha}.
\end{equation*}
Thus
\[
\nu(c,d)= \dfrac{\tau(c-d)}{\tau(c \lambda_2 -d)} = \dfrac{\tau(c-d)}{-\tau(c \lambda_N -d)}
\]
and taking into account that
\[
\dfrac{\hbox{d }}{\hbox{d }\alpha}(c\lambda-d)=2\dfrac{2\lambda -\lambda_2-\lambda_N}{(\lambda_2-\lambda_N-2\alpha)^2}=
2\dfrac{c\lambda -d}{(\lambda_2-\lambda_N-2\alpha)},
\]

\begin{equation*}
\dfrac{\hbox{d }\nu(c,d)}{\hbox{d }\alpha}=
\dfrac{-2\tau(c-d)}{\tau(c \lambda_2 -d)(\lambda_2-\lambda_N-2\alpha)} \left[\dfrac{c-d}{\sqrt{(c-d)^2-1}}
- \dfrac{c\lambda_2-d}{\sqrt{(c\lambda_2-d)^2-1}}  \right]>0.
\end{equation*}
Then $\nu(c,d)$ is increasing with $\alpha$ and the minimum value is obtained for $\alpha=0$.
%\end{proof}
\EndProof 

%% file: proof_convergenceSwitching.tex
First of all, let us state the notation we will follow
along the proof.
%For any weight matrix $\textbf{A}(n)$
%we denote its eigenvalues and eigenvectors by
%$\lambda_i(n)$ and $\textbf{v}_i(n),\ i=1,\ldots,N,$ respectively.
For any weight matrix $\textbf{A}(n)$
we denote its eigenvectors by
$\textbf{v}_i(n),\ i=1,\ldots,N$.
Let us denote $\textbf{V}(n) = [\textbf{v}_1(n),\ldots,\textbf{v}_N(n)]$
the matrix with all the eigenvectors of $\textbf{A}(n)$.
Thus, $\textbf{A}(n)\textbf{V}(n) = \textbf{V}(n)\textbf{D}(n),$ with
$\textbf{D}(n) = \hbox{diag}(\lambda_1(n),\ldots,\lambda_N(n)).$
Since $\textbf{A}(n)$ is symmetric, it is diagonalizable and
we can choose the base of eigenvectors in such a way that $\textbf{V}(n)$ is orthogonal.
Therefore,
$\textbf{v}_1(n)^T\textbf{v}_i(n) = 0, \forall i=2,\ldots,N,$
and
$\textbf{v}_1(n) = \textbf{1}/\sqrt{N} = \textbf{v}_1$,
for all $n$.
%We denote
%\begin{equation}
%\lambda_{\max} = \max_{n} \max_{i=2,\ldots,N} \lambda_i(n)<1, \hbox{ and }
%\lambda_{\min} = \min_{n} \min_{i=2,\ldots,N} \lambda_i(n)>-1,
%\end{equation}

Let $\textbf{Q}(n)= \textbf{A}(n)- \frac{1}{N}\textbf{1}\textbf{1}^T$,
whose eigenvalues are $0,$ with $\textbf{v}_1(n)=\textbf{1}/\sqrt{N}$ its corresponding eigenvector,
and $\lambda_2(n), \ldots, \lambda_N(n),$ with the same eigenvectors as $\textbf{A}(n)$.
Taking all of this into account it is easy to see that
$\textbf{1}\textbf{1}^T(\textbf{x}(0)-(\textbf{1}^T\textbf{x}(0))\textbf{v}_1) = 0,$
and
\begin{equation}
\textbf{A}(n) (\textbf{x}(n) - (\textbf{1}^T\textbf{x}(0))\textbf{v}_1) =
\textbf{Q}(n) (\textbf{x}(n) - (\textbf{1}^T\textbf{x}(0))\textbf{v}_1).
\end{equation}
Given two consecutive matrices, $\textbf{Q}(n)$ and $\textbf{Q}(n-1),$
let $\textbf{P}(n)$ be the matrix such that $\textbf{V}(n-1) = \textbf{V}(n)\textbf{P}(n),$
that is, the matrix that changes from the base of eigenvectors of $\textbf{Q}(n-1)$ to the
base of eigenvectors of $\textbf{Q}(n).$
In a similar way, $\textbf{R}(n)$ will be such that $\textbf{V}(n-2) = \textbf{V}(n)\textbf{R}(n).$
The orthogonality of $\textbf{V}(n)$,
implies that the matrices $\textbf{P}(n) = \textbf{V}(n)^{-1}\textbf{V}(n-1)$ and
$\textbf{R}(n) = \textbf{V}(n)^{-1}\textbf{V}(n-2)$ are also orthogonal,
and $\|\textbf{P}(n)\|_2=\|\textbf{R}(n)\|_2=1$.
%By Assumption \ref{TimeVaryingStochasticMatrices} $\textbf{A}(n)$ is symmetric,
%which implies that the matrices $\textbf{P}(n)$ and $\textbf{R}(n)$ have all the elements
%in the first row and the first column equal to zero but the elements $p_{11}(n)=q_{11}(n)=1.$

%and therefore $P_n(\textbf{A})(\textbf{x}(0)-\textbf{v}_1) = p_n(\textbf{Q})(\textbf{x}(0)-\textbf{v}1).$

Recalling the Chebyshev recurrence \eqref{ChebyshevRecurrenceSwitch},
we define the error at iteration $n$
by $\textbf{x}(n) - (\textbf{1}^T\textbf{x}(0))\textbf{v}_1.$
The equivalence
\begin{equation}
\label{RecurrenceV1}
%\textbf{v}_1 = \frac{2(c\textbf{Q}(n)-d\textbf{I})\textbf{v}_1 - \textbf{v}_1}{T_n(c-d)}.
\textbf{v}_1 = 2\dfrac{T_{n}(c-d)}{T_{n+1}(c-d)}(c \textbf{A}(n) -d \textbf{I})\textbf{v}_1
-\dfrac{T_{n-1}(c-d)}{T_{n+1}(c-d)} \textbf{v}_1.
\end{equation}
allows us to express the error by $\textbf{e}(n)/T_n(c-d)$,
%$\textbf{e}(n) = \textbf{x}(n) - (\textbf{1}^T\textbf{x}(0))\textbf{v}_1,$
%and recalling the Chebyshev recurrence \eqref{ChebyshevRecurrenceSwitch2},
%the error can be expressed as $\textbf{e}(n) = \textbf{u}(n)/T_n(c-d)$,
with $\textbf{e}(0) = \textbf{x}(0) - (\textbf{1}^T\textbf{x}(0))\textbf{v}_1,$
$\textbf{e}(1) = (c\textbf{Q}(1)-d\textbf{I})\textbf{e}(0)$ and
\begin{equation}
\label{RecurrenceMultTn}
\textbf{e}(n) = 2(c\textbf{Q}(n)-d\textbf{I})\textbf{e}(n-1) - \textbf{e}(n-2).
\end{equation}
%In the previous equivalence we have used the following equivalence
%\begin{equation}
%\label{RecurrenceV1}
%\textbf{v}_1 = 2\dfrac{T_{n}(c-d)}{T_{n+1}(c-d)}(c \textbf{A}(n) -d \textbf{I})\textbf{v}_1
%-\dfrac{T_{n-1}(c-d)}{T_{n+1}(c-d)} \textbf{v}_1.
%\end{equation}

Each vector $\textbf{e}(n)$ can be expressed as a linear combination of the
eigenvectors of $\textbf{Q}(n),$
\begin{equation}
\label{combinationEig}
\textbf{e}(n) = \sum_{i=1}^N \alpha_i(n) \textbf{v}_i(n) = \textbf{V}(n) \bm \alpha(n).
\end{equation}
Replacing $\textbf{e}(n)$ by \eqref{combinationEig} in \eqref{RecurrenceMultTn},
\begin{equation}
\begin{split}
\textbf{e}(n) &= 2(c\textbf{Q}(n)-d\textbf{I})\textbf{V}(n-1)\bm\alpha(n-1) - \textbf{V}(n-2)\bm\alpha(n-1)\\
&= 2(c\textbf{Q}(n)-d\textbf{I})\textbf{V}(n)\textbf{P}(n)\bm\alpha(n-1) - \textbf{V}(n)\textbf{R}(n)\bm\alpha(n-2)\\
&= 2\textbf{V}(n)(c\textbf{D}(n)-d\textbf{I})\textbf{P}(n)\bm\alpha(n-1) -
\textbf{V}(n)\textbf{R}(n)\bm\alpha(n-2)\\
&= \textbf{V}(n)[2(c\textbf{D}(n)-d\textbf{I})\textbf{P}(n)\bm\alpha(n-1) -
\textbf{R}(n)\bm\alpha(n-2)]=
\textbf{V}(n)\bm\alpha(n).
\end{split}
\end{equation}
Therefore, the vectors $\bm\alpha(n)$ satisfy the recurrence
\begin{equation}
\bm\alpha(n) = 2(c\textbf{D}(n)-d\textbf{I})\textbf{P}(n)\bm\alpha(n-1) -
\textbf{R}(n)\bm\alpha(n-2),
\end{equation}
with $\bm\alpha(0) = \bm\alpha(1)$.

Taking spectral norms,
\begin{equation}
\begin{split}
\|\bm\alpha(n)\|_2 &= \|2(c\textbf{D}(n)-d\textbf{I})\textbf{P}(n)\bm\alpha(n-1) -
\textbf{R}(n)\bm\alpha(n-2)\|_2 \leq\\
&\leq 2\|(c\textbf{D}(n)-d\textbf{I})\|_2 \|\textbf{P}(n)\|_2 \|\bm\alpha(n-1)\|_2
+ \|\textbf{R}(n)\|_2\|\bm\alpha(n-2)\|_2 \leq\\
&\leq (2\max_{i}|c\lambda_i(n)-d| \|\bm\alpha(n-1)\|_2 + \|\bm\alpha(n-2)\|_2).
\end{split}
\end{equation}
%with
%\begin{equation}
%K = \max_n \{\|\textbf{P}(n)\|,\|\textbf{R}(n)\|,1\}.
%\end{equation}
By Lemma \ref{LemmaSwitch}
we can bound the norm of $\|\bm\alpha(n)\|$ by
\begin{equation}
\|\bm\alpha(n)\| \leq \kappa_1(x_{\max})^n \|\bm\alpha(0)\|,
\end{equation}
where the parameter $x_{\max}$ in this case is
\begin{equation}
\begin{split}
x_{\max} &= \max_{n} \max_{i=2,\ldots,N} |c\lambda_i(n)-d| =
\max_{n} \{ |c\lambda_2(n)-d| , |c\lambda_N(n)-d| \}=\\
&=\max \{ |c\lambda_{\max}-d|, |c\lambda_{\min} - d| \}.
\end{split}
\end{equation}

Therefore, in order to make the error go to zero we require
that
\begin{equation}
\lim_{n\to\infty}\frac{\kappa_1(x_{\max})^n}{T_n(c-d)}= 0.
\end{equation}
Using \eqref{direct_Expression}
\begin{equation}
\frac{\kappa_1(x_{\max})^n}{T_n(c-d)}= \frac{\kappa_1(x_{\max})^n\tau(c-d)^n}{1+\tau(c-d)^{2n}},
\end{equation}
which goes to zero if $\kappa_1(x_{\max})\tau(c-d)<1.$
When this happens
$\lim_{n\to\infty} \textbf{x}(n)=(\textbf{1}^T\textbf{x}(0)/\textbf{1}^T\textbf{1})\textbf{1},$
and the consensus is achieved.

%The denominators tends to $1$ with $n,$ then
%the consensus is achieved if $\kappa\tau(c-d) < 1$.
%It is equivalent to
%\begin{equation}
%\begin{split}
%(c-d)^2 - (c\lambda-d)^2 &> 1\\
%c(c-d + c\lambda-d) &> \frac{1}{(1-\lambda)}
%c^2(1+\lambda-(\lambda_M+\lambda_m)) &> \frac{1}{(1-\lambda)}
%\end{split}
%\end{equation}

\EndProof

%% file: main.bbl
\begin{thebibliography}{10}

\bibitem{JCM-DCH:02}
J.C. Mason and D.~C. Handscomb.
\newblock {\em Chebyshev Polynomials}.
\newblock Chapman and Hall, 2002.

\bibitem{CN-HHC-AMD:02}
C.~Naik, H.H. Cartensen, and A.~M. Dean.
\newblock Reaction rate representation using chebyshev polynomials.
\newblock In {\em Western States Section 2002 Spring Meeting of the Combustion
  Institute}, pages 5714--5719, March 2002.

\bibitem{DLR-DS-JM:98}
D.~L. Richardson, D.~Schmidt, and J.~Mitchell.
\newblock Improved chebyshev methods for the numerical integration of
  first-order differential equations.
\newblock In {\em Spaceflight mechanics 1998; Proceedings of the AAS/AIAA Space
  Flight Mechanics Meeting}, pages 1533--1544, February 1998.

\bibitem{JGV-BPS:04}
J.G. Verwer and B.P. Sommeijer.
\newblock An implicit–-explicit runge–-kutta–-chebyshev scheme for
  diffusion–-reaction equations.
\newblock {\em SIAM Journal on Scientific Computing}, 25(5):1824--1835, May
  2004.

\bibitem{SO-CHC-HAC:11}
S.~Ozera, C.~H. Chenb, and H.~A. Cirpanc.
\newblock A set of new chebyshev kernel functions for support vector machine
  pattern classification.
\newblock {\em Pattern Recognition}, 44(7):1435--1447, July 2011.

\bibitem{FB-JC-SM:09}
F.~Bullo, J.~Cort\'es, and S.~Mart{\'\i}nez.
\newblock {\em Distributed Control of Robotic Networks}.
\newblock Applied Mathematics Series. Princeton University Press, 2009.
\newblock Electronically available at http://coordinationbook.info.

\bibitem{WR-RWB:08}
W.~Ren and R.~W. Beard.
\newblock {\em Distributed Consensus in Multi-vehicle Cooperative Control}.
\newblock Communications and Control Engineering. Springer-Verlag, London,
  2008.

\bibitem{ROS-RMM:04}
R.~Olfati-Saber and R.~M. Murray.
\newblock Consensus problems in networks of agents with switching topology and
  time-delays.
\newblock {\em {IEEE} Transactions on Automatic Control}, 49(9):1520--1533,
  September 2004.

\bibitem{AJ-JL-ASM:03}
A.~Jadbabaie, J.~Lin, and A.S. Morse.
\newblock Coordination of groups of mobile autonomous agents using nearest
  neighbor rules.
\newblock {\em {IEEE} Transactions on Automatic Control}, 48(6):988--1001, June
  2003.

\bibitem{MZ-SM:10}
M.~Zhu and S.~Mart{\'\i}nez.
\newblock Discrete-time dynamic average consensus.
\newblock {\em Automatica}, 46(2):322--329, February 2010.

\bibitem{MM-DS-JP-SHL-RMM:07}
M.~Mehyar, D.~Spanos, J.~Pongsajapan, S.~H. Low, and R.~M. Murray.
\newblock Asynchronous distributed averaging on communication networks.
\newblock {\em IEEE/ACM Transactions on Networking}, 15(3):512--520, June 2007.

\bibitem{LX-SB-SL:05}
L.~Xiao, S.~Boyd, and S.~Lall.
\newblock A scheme for robust distributed sensor fusion based on average
  consensus.
\newblock In {\em International Conference on Information Processing in Sensor
  Networks}, pages 63--70, Los Angeles, April 2005.

\bibitem{GFT-AB-MG:06}
G.~Ferrari-Trecate, A.~Buffa, and M.~Gati.
\newblock Analysis of coordination in multi-agent systems through partial
  difference equations.
\newblock {\em {IEEE Transactions on Automatic Control}}, 51(6):1058--1063,
  June 2006.

\bibitem{EL-FG-SZ:10}
E.~Lovisari, F.~Garin, and S.~Zampieri.
\newblock A resistance-based approach to performance analysis of the consensus
  algorithm.
\newblock In {\em IEEE Int. Conference on Decision and Control}, pages
  5714--5719, December 2010.

\bibitem{SP-BB-AEA:10}
S.~Patterson, B.~Bamieh, and A.~E. Abbadi.
\newblock Convergence rates of distributed average consensus with stochastic
  link failures.
\newblock {\em {IEEE} Transactions on Automatic Control}, 55(4):880--892, April
  2010.

\bibitem{JZ-QW:09}
J.~Zhou and Q.~Wang.
\newblock Convergence speed in distributed consensus over dynamically switching
  random networks.
\newblock {\em {A}utomatica}, 45(6):1455--1461, June 2009.

\bibitem{MZ-SM:08b}
M.~Zhu and S.~Mart{\'\i}nez.
\newblock On the convergence time of asynchronous distributed quantized
  averaging algorithms.
\newblock {\em {IEEE Transactions on Automatic Control}}, 56(2):386--390,
  February 2011.

\bibitem{FJ-LW:09}
F.~Jiang and L.~Wang.
\newblock Finite-time information consensus for multi-agent systems with fixed
  and switching topologies.
\newblock {\em Physica D}, 238(16):1550--1560, August 2009.

\bibitem{JC:06}
J.~Cort{\'e}s.
\newblock Finite-time convergent gradient flows with applications to network
  consensus.
\newblock {\em {A}utomatica}, 42(11):1993--2000, November 2006.

\bibitem{LW-FX:10}
L.~Wang and F.~Xiao.
\newblock Finite-time consensus problems for networks of dynamic agents.
\newblock {\em {IEEE} Transactions on Automatic Control}, 55(4):950--955, April
  2010.

\bibitem{CKK-XG:09}
C.K. Ko and X.~Gao.
\newblock On matrix factorization and finite-time average-consensus.
\newblock In {\em IEEE Int. Conference on Decision and Control}, pages
  5798--5803, December 2009.

\bibitem{ZJ-RMM:06}
Z.~Jin and R.M. Murray.
\newblock Multi-hop relay protocols for fast consensus seeking.
\newblock In {\em IEEE Int. Conference on Decision and Control}, pages
  1001--1006, December 2006.

\bibitem{DY-SX-HZ-YC:10}
D.~Yuan, S.~Xu, H.~Zhaoa, and Y.~Chub.
\newblock Accelerating distributed average consensus by exploring the
  information of second-order neighbors.
\newblock {\em Physica Letters A}, 37(4):2438–2445, May 2010.

\bibitem{LX-SB:04}
L.~Xiao and S.~Boyd.
\newblock Fast linear iterations for distributed averaging.
\newblock {\em {S}ystems and {C}ontrol {L}etters}, 53:65--78, September 2004.

\bibitem{YK-DWG-IP:09}
Y.~Kima, D.W. Gub, and I.~Postlethwaite.
\newblock Spectral radius minimization for optimal average consensus and output
  feedback stabilization.
\newblock {\em {A}utomatica}, 45(6):1379--1386, June 2009.

\bibitem{BJ-MJ:08}
B.~Johansson and M.~Johansson.
\newblock Faster linear iterations for distributed averaging.
\newblock In {\em {$17^{th}$ IFAC World Congress}}, July 2008.

\bibitem{SS-CNH:07b}
S.~Sundaram and C.~N. Hadjicostis.
\newblock Finite-time distributed consensus in graphs with time-invariant
  topologies.
\newblock In {\em American Control Conference}, pages 711--716, New York, July
  2007.

\bibitem{YY-GBS-LS-JG:09}
Y.~Yuan, G.~Stan, L.~Shi, and J.~Gon{\c{c}}alves.
\newblock Decentralised final value theorem for discrete-time lti systems with
  application to minimal-time distributed consensus.
\newblock In {\em IEEE Int. Conference on Decision and Control}, pages
  2664--2669, December 2009.

\bibitem{EK-PF:09}
E.~Kokiopoulou and P.~Frossard.
\newblock Polynomial filtering for fast convergence in distributed consensus.
\newblock {\em IEEE Transactions on Signal Processing}, 57(1):342–354, January
  2009.

\bibitem{TCA-BNO-MJC:09}
T.C. Aysal, B.~Oreshkin, and M.J. Coates.
\newblock Accelerated distributed average consensus via localized node state
  prediction.
\newblock {\em IEEE Transactions on Signal Processing}, 57(4):1563--1576, April
  2009.

\bibitem{BO-MC-MR:10}
B.~Oreshkin, M.~J Coates, , and M.~Rabbat.
\newblock Optimization and analysis of distributed averaging with short node
  memory.
\newblock {\em IEEE Transactions on Signal Processing}, 58(5):2850--2865, May
  2010.

\bibitem{SM-BG-MHS:98}
B.~Ghosh S.~Muthukrishnan and M.~H. Schultz.
\newblock First- and second-order diffusive methods for rapid, coarse,
  distributed load balancing.
\newblock {\em Theory of Computing Systems}, 31(4):331--354, December 1998.

\bibitem{EM-JIM-CS:11}
E.~Montijano, J.~I. Montijano, and C.~Sagues.
\newblock {Fast Distributed Consensus with Chebyshev Polynomials}.
\newblock In {\em American Control Conference}, 2011.
\newblock (to appear).

\bibitem{RLGC-AR-NRJ:11}
R.~L.~G. Cavalgante, A.~Rogers, and N.~R. Jennings.
\newblock {Consensus Acceleration in Multiagent Systems with the Chebyshev
  Semi-iterative Method}.
\newblock In {\em {$10^{th}$ Int. Conf. on Autonomous Agents and Multiagent
  Systems}}, 2011.
\newblock (to appear).

\bibitem{RPA:92}
R.~P. Agarwal.
\newblock {\em Difference equations and inequalities, Theory, Methods and
  Applications}.
\newblock Dekker, 1992.

\end{thebibliography}
